\RequirePackage{fix-cm}
\documentclass[smallextended]{svjour3}
\usepackage{graphicx,placeins}
\usepackage[utf8]{inputenc}

\smartqed

\usepackage{latexsym,amsmath,amssymb,hyperref}
\usepackage{amsfonts,xcolor,hyperref,appendix}
\usepackage{amsthm,comment} 

\begin{document}
\bibliographystyle{spmpsci}  

\title{Elementary Integral Series for Heun Functions}
\subtitle{With an Application to Black-Hole Perturbation Theory}

\author{Pierre-Louis Giscard        \and
        Aditya Tamar 
}

\institute{Pierre-Louis Giscard \at
              Univ. Littoral C\^ote d’Opale, UR 2597, LMPA, Laboratoire de Mathématiques Pures et Appliqu\'ees Joseph Liouville, F-62100 Calais, France. \\
              \email{giscard@univ-littoral.fr}          
           \and
           Aditya Tamar \at
              Independent Researcher, Delhi, India. \\
              \email{adityatamar@gmail.com}  
}

\date{Received: date / Accepted: date}

\maketitle

\begin{abstract}
Heun differential equations are the most general second order Fuchsian equations with four regular singularities. 
An explicit integral series representation of Heun functions involving only elementary integrands has hitherto been unknown and noted as an important open problem in a recent review. 
We provide explicit integral representations of the solutions of all equations of the Heun class: general, confluent, bi-confluent, doubly-confluent and triconfluent, with integrals involving only rational functions and exponential integrands. All the series are illustrated with concrete examples of use. These results stem from the technique of path-sums, which we use to evaluate the path-ordered exponential of a variable matrix chosen specifically to yield Heun functions. 
We demonstrate the utility of the integral series by providing the first representation of the solution to the Teukolsky radial equation governing the metric perturbations of rotating black holes that is convergent everywhere from the black hole horizon up to spatial infinity. 
\end{abstract}

\section{Introduction}
\label{intro}
The study of Heun equations has generated significant interest in both mathematics and physics lately. From a mathematical standpoint, recent results have uncovered relation between Heun equations other equations of paramount importance for physics. 
For example, it was found by means of antiquantisation procedures \cite{Slav1} and monodromy preserving transformations \cite{Takem1} that the Heun equations share a bijective relationship with Painlev{\'e} equations \cite{Slav1,Slav2,Slav4}. This permitted in-depth studies on the integral symmetry properties of equations of the Heun class \cite{Slav3} and to determine generating polynomial solutions of the Heun equation by formulating a Riemann-Hilbert problem for the Heun function \cite{RieHibH}. The reduction of certain Heun equations under non-trivial substitutions to hypergeometric equations has also been possible by means of pull-back transformations based on Belyi coverings \cite{BelyiHeun} and polynomial transformations \cite{Maier1,Maier2}.

In contrast, in spite of the increasing use of Heun functions in physics (in quantum optics \cite{Xie2010,Moham}, condensed matter physics \cite{Crampe,Dorey}, quantum computing \cite{QuantComp}, two-state problems \cite{QTS1,QTS2} and more \cite{Hort1}), few studies \cite{Hort1,Hort2} have specifically focused on determining their properties most relevant to physical applications. For example, the lack of integral expansions of these functions involving only elementary integrands has been clearly identified as a major obstacle when extracting physical meaning from the mathematical treatment of black holes quasinormal modes \cite{Hort1,Hort2}, yet remains unaddressed in the mathematical literature.
The present works tackles this issue by determining a novel integral representation of the Heun equations involving elementary functions that is tailored to physical applications. In particular, we demonstrate the applicability of the novel integral representation to the Teukolsky equation \cite{TeukEqn} that governs the metric perturbations of rotating black holes and further explore which physical observables pertinent to black hole perturbation theory can be obtained from the integral form. The present progress in integral representation is enabled by the method of path-sum \cite{Giscard2015}, which generates the linear Volterra integral equation of the second kind satisfied by any function involved a system of coupled linear differential equations with variable coefficients.

This paper is organised as follows. In Section \ref{sec:2} we give the minimal necessary background on Heun equations. This section concludes in \S\ref{sec:IntRep} with a review of existing integral representations of Heun functions and their major drawback as noted in the recent mathematical-physics literature. 
 Section \ref{sec:ExplicitRes} is a self-contained presentation of the novel, elementary integral representations of all functions of Heun class, illustrated with concrete examples. This section contains none of the proofs, all of which are deferred to Appendix~\ref{AppendixProofs}. Then, in Section~\ref{sec:black hole} we give the elementary integral series representation of the solution to the Teukolsky radial equation. This representation is the first one to be convergent from the black hole horizon up to spatial infinity. This stands in contrast to the state-of-the-art MST formalism \cite{Mano1997}, that uses \textit{two} hypergeometric series (one convergent at the horizon and the other at infinity) that must then be matched after an analytic continuation procedure. This last step requires the introduction of an auxiliary parameter lacking physical correspondence, at the very least obscuring the physical picture. The convergence of the integral series over the entire domain from the black hole horizon up to spatial infinity therefore alleviates the need for such parameters lacking physical correspondence when calculating solutions of the Teukolsky radial equation. These solutions are of primary importance for computing quantities of physical interest such as gravitational wave fluxes \cite{Fujita2004} and quasinormal modes \cite{Zhang2013}.
 We conclude in \S\ref{sec:conclusion} with a brief discussion of the novel integral series and future prospects of the method of path-sum from which they stem for solving the coupled system of Teukolsky angular and radial equations.

\section{Heun Differential Equations}
\label{sec:2}
\subsection{Mathematical Context}
The most general linear, homogenous, second order differential equation with polynomial coefficients is given by the Fuchsian equation \cite{SpecFunc} which has the following form
\begin{equation*}
P(z)\frac{d^2y(z)}{dz^2} + Q(z)\frac{dy(z)}{dz} + R(z)y(z) = 0,~ z \in \mathbb{CP}^1,
\end{equation*} 
where $\mathbb{CP}^1$ is the Riemann sphere. In the above equation, if the function $K_{QP} = Q(z)/P(z)$ has a pole of at most first order and $K_{RP} = R(z)/P(z)$ has a pole of at most second order at some singularity $z=z_0$, then $z_0$ is called a \textit{Fuchsian} singularity, otherwise it is an \textit{irregular} singularity. 
The above equation is a \textit{Fuchsian equation} if all its singularities are Fuchsian singularities. Now, any Fuchsian equation with exactly four singular points can be mapped onto a Heun equation \cite{HeunOrig} by transformation in dependent or independent variables. These transformations are called s-homotopic and M{\"o}bius transformations respectively. The Heun equation is a straightforward generalisation of the hypergeometric equation, a Fuchsian equation with exactly three singular points \cite{SpecFunc}.

\subsection{General Heun Equation}
As mentioned in the Introduction, the Heun differential equation is the most general Fuchsian equation with four regular singularities. The canonical form of the equation, also known as the General Heun Equation (GHE)  is given by the following equation and conditions: 
\begin{equation}
\frac{d^2y(z)}{dz^2} + \bigg[ \frac{\gamma}{z} + \frac{\delta}{z - 1} + \frac{\epsilon}{z - t} \bigg]\frac{dy(z)}{dz} + \frac{\alpha\beta z - q}{z(z-1)(z-t)}y(z) = 0, \label{eq:Heun}
\end{equation}
where $q \in \mathbb{C}$ is called the \textit{accessory parameter}. The corresponding Riemann-P symbol is as follows: 
\begin{equation*}
\begin{pmatrix}
0 & 1 & a & \infty \\
0 & 0 & 0 & \alpha&;z \\
1-\gamma & 1-\delta & 1 - \epsilon & \beta
\end{pmatrix}
\end{equation*}
where the parameters satisfy the Fuch's condition:
\begin{equation*}
1 + \alpha + \beta = \gamma + \delta + \epsilon
\end{equation*}
The GHE has four singular points at $z = 0,1,t,\infty$. Concerning its solutions, Maier, completing a task initiated by Heun \cite{Heun1888} himself has shown that solutions of the GHE have Coxeter group $D_4$ as their automorphism group \cite{Maier192}. This means that 192 solutions can be generated using the symmetries of $D_4$, much more than the 24 solutions of the Gauss Hypergeometric equation determined by Kummer \cite{Whittaker}. We refer the reader to \cite{Maier192} for the complete list of solutions and their relations as well as to \cite{SpecFunc} for a further discussion of their properties. 

For specific parameter values the Heun equation reduces to other well-known equations of importance: e.g. setting $\epsilon = 0, \gamma = \delta = 1/2$ yields the Mathieu equation, which has found widespread applicability in the theoretical and experimental study of vibration phenomenon \cite{Mathieu,Mathieu5}, electromagnetic scattering from elliptic waveguides \cite{Mathieu1,Mathieu2,Mathieu3}, ion traps in mass spectrometry \cite{Mathieu4}, stability of floating ships \cite{Mathieu6}. Furthermore, the confluent form of the Heun equation has found wide ranging applications in quantum particle confinement and interaction potentials \cite{ConfPot,IntPot} and in the Stark effect \cite{Slav5,SpecFunc}. 

\subsection{Confluent Heun Equations}
\label{sec:2.1}
The GHE contains 4 regular singularities. If we apply a \textit{confluence} procedure to two of its singularities such that we get an irregular singularity, we call the resultant equation a confluent Heun equation (CHE). The CHE contains at least one irregular singular point besides the regular singular points. We can construct local solutions in the vicinity of this irregular singular points by the means of (generally divergent) Thom{\'e} series \cite{SpecFunc}.  The number of parameters in the CHE are reduced by one. Thus by applying the confluence procedure laid out in \cite{SpecFunc} to the singularities at $z = t$ and $z = \infty$ in equation 2, we get the CHE:
\begin{equation}\label{eq:HeunConfluent}
\frac{d^2y(z)}{dz^2} + \bigg[ \frac{\gamma}{z} + \frac{\delta}{z - 1} + \epsilon \bigg]\frac{dy(z)}{dz} + \frac{\alpha z - q}{z(z-1)}y(z) = 0.
\end{equation}
By continuing application of the confluence procedure, we obtain the bi-confluent Heun equation
\begin{equation}\label{eq:HeunBiConfluent}
\frac{d^2y(z)}{dz^2} + \bigg[ \frac{\gamma}{z} + \delta + \epsilon z \bigg]\frac{dy(z)}{dz} + \frac{\alpha z - q}{z}y(z) = 0,
\end{equation}
and related doubly-confluent Heun equation
\begin{equation}\label{eq:HeunDoublyConfluent}
\frac{d^2y(z)}{dz^2} + \bigg[ \frac{\delta}{z^2}+\frac{\gamma}{z}  + 1 \bigg]\frac{dy(z)}{dz} + \frac{\alpha z - q}{z^2}y(z) = 0,
\end{equation}
as well as the triconfluent Heun equation 
\begin{equation}\label{eq:HeunTriConfluent}
\frac{d^2y(z)}{dz^2} + \bigg[ \gamma + \delta z + \epsilon z^2 \bigg]\frac{dy(z)}{dz} + (\alpha z - q)y(z) = 0.
\end{equation}
We refer the reader to \cite{SpecFunc} for further general informations on these functions. 

\subsection{Integral representations of Heun functions}
\label{sec:IntRep}
Erd{\'e}lyi was the first to give an integral equation relating the values taken at two points by a general Heun function  \cite{Erdelyi}. His equation, a Fredholm integral equation, involves an hypergeometric kernel and can be used to obtain a series representation of Heun functions as sums of hypergeometric functions with coefficients determined via recurrence relations. Applications of this result in the special cases of Mathieu and Lam{\'e} equations were discussed by Sleeman \cite{Sleeman}. Naturally, since Erd{\'e}lyi's breakthrough many mathematical works on Heun equations were concerned with integral transformations involving Heun functions. In particular, based on the work of Carlitz\cite{Carlitz}, Valent found an integral transform for the Heun equation in terms of Jacobi polynomials \cite{Valent}; Ishkhanyan gave expansions of the confluent Heun functions involving incomplete beta functions \cite{Ishkhanyan_2005}; El Jaick and coworkers \cite{El_Jaick_2011} provided novel transformations and classified expansions for Heun functions involving hypergeometric kernels; and Takemura found an elliptic transformation relating Heun's functions for different parameters based on the Weierstrass sigma function \cite{Takem2}. This brief list of contributions is far from exhaustive,  we refer to the recent review \cite{Hort1} for more details.

The common feature of all of these integral transforms is that they contain higher transcendental functions which makes them physically opaque and of limited use for practical calculations. In addition, the resulting series representations for the Heun functions have insufficient radiuses of convergence \cite{Cook2016} causing difficulties for black hole perturbation theory (see Section~\ref{sec:5}).
These issues were noted in the recent review \cite{Hort1} 
on Heun's functions, the current state of research on this being described as follows :\\[-1.1em]

\textit{``No example has been given of a solution of Heun's equation expressed in the form of a definite integral or contour integral involving only functions which are, in some sense, simpler.[...] This statement does not exclude the possibility of having an infinite series of integrals with `simpler' integrands''.}\\[-1.1em]

In this work, we constructively prove the existence of such a representation for all types of Heun's functions and for all parameters, in the form of infinite series of integrals whose integrands involve only rational functions and exponentials of polynomials. Furthermore, we show that the series converges everywhere except at the singular points of the Heun function. We show that any Heun function, general or (bi-, doubly-, tri-)confluent, is a sum of exactly two functions each of which satisfy a linear Volterra equation of the second kind with explicitely identified elementary kernels. In particular, any Heun function $H(z)$ itself satisfies a linear integral Volterra equation of the second kind with such an elementary kernel if either there is at least one non-singular point $z_0\in\mathbb{R}$ where $H(z_0)=H'(z_0)$ or there is a point where $H(z_0)=0$.

\section{Elementary integral series for all types of  Heun functions}\label{sec:ExplicitRes}
Owing to the emphasis of the present work on  concrete results and a physical application, all the technical mathematical proofs are deferred to Appendix~\ref{AppendixProofs}.
\subsection{Notation}
The $\ast$ notation is useful to denote iterated integrals. Let $K(z,z_0)$ be a function of two variables that is continuous over $]z_0,z[$. We denote $K(z,z_0)=K^{\ast1}(z,z_0)$ and, for any integer $n>1$, 
$$
K^{\ast n}(z,z_0)=\int_{z_0}^{z} K^{\ast (n-1)}(z,\zeta_1)K(\zeta_1,z_0)d\zeta.
$$
In other terms $K^{\ast n}$ is the Volterra composition \cite{Volterra1924} of $K$ with itself $n$-times. The only type of integral series that is required to present all results of this section is the following
\begin{align*}
G(z,z_0)&:= \sum_{n=1}^\infty K^{\ast n}(z,z_0),\\
&=K(z,z_0)+\int_{z_0}^z K(z,\zeta_1)K(\zeta_1,z_0)d\zeta_1\nonumber\\
&\hspace{5mm}+\int_{z_0}^z \int_{\zeta_1}^z K(z,\zeta_2)K(\zeta_2,\zeta_1)K(\zeta_1,z_0)d\zeta_2d\zeta_1\,+\cdots, \nonumber
\end{align*}
see also Eq.~(\ref{Gexplicitform}). In the appendix, we show that once $z_0$ is fixed, the above series converges over any subinterval of $\mathbb{R}$ which does not contain a singularity of $K(z,z_0)$. 
A bound on the convergence speed of the series is also provided. 

The function $G(z,z_0)$ defined above, is solution to the linear Volterra integral equation of the second kind
\begin{equation}\label{VolterraG}
G(z,z_0)=K(z,z_0)+\int_{z_0}^z K(z,\zeta)G(\zeta,z_0)d\zeta,
\end{equation}
or, in $\ast$ notation, $G=K+K\ast G$. Thus, the function $G$ can either be evaluated from the integral series or by solving the above Volterra equation.

\subsection{Results}
We emphasize that all results stated remain valid for complex parameter values. This is crucial notably when forming solutions of the Teukolsky equation in the study of quasinormal modes, for which the frequency parameter takes complex values (see Eq.~\ref{eq:18}).

\begin{corollary}[General Heun Equation]\label{GenHeunCorrolary}
Let $H_G(z)$ be solution of the General Heun Equation,
 \begin{equation*}
\frac{d^2H_G(z)}{dz^2} + \bigg[ \frac{\gamma}{z} + \frac{\delta}{z - 1} + \frac{\epsilon}{z - t} \bigg]\frac{dH_G(z)}{dz} + \frac{\alpha\beta z - q}{z(z-1)(z-t)}H_G(z) = 0,
\end{equation*}
with initial conditions $H_G(z_0)=H_0$ and $\dot{H}_G(z_0)=H'_0$, assuming that $z_0\in\mathbb{R}$ is not a singular point of $H_G$. Denote $I$ the largest real interval that contains $z_0$ and does not contain any singular point of $H_G$. Then, for any $z\in I$,  
\begin{align*}
  H_G(z)&=H_0+H_0\int_{z_0}^z \!G_{1}(\zeta,z_0)d\zeta
+(H'_0-H_0)\left(\!e^{z-z_0}-1+\!\int_{z_0}^z\!(e^{z-\zeta}-1)G_{2}(\zeta,z_0)d\zeta\right),
  \end{align*}
where  $G_{i}=\sum_{n=1}^\infty K_{i}^{\ast n}$ and  
   \begin{align*}
K_{1}(z,z_0)&=\\
&\hspace{-15mm}1+ e^{-z}\int_{z_0}^z \Big\{\frac{\zeta_1^{\gamma } (\zeta_1-1)^{\delta }(t-\zeta_1)^{\epsilon }}{z^{\gamma }
   (z-1)^{\delta } (t-z)^{\epsilon }}e^{\zeta_1} \left(\frac{q-\alpha  \beta 
   \zeta _1}{\left(\zeta
   _1-1\right) \zeta _1
   \left(\zeta
   _1-t\right)}-\frac{\epsilon
   }{t-\zeta _1}-\frac{\gamma }{\zeta
   _1}-\frac{\delta }{\zeta
   _1-1}-1\right)\Big\}d\zeta_1,\\
K_{2}(z,z_0)&=\left(\frac{q-\alpha  \beta  z}{(z-1) z
   (z-t)}-\frac{\epsilon
   }{t-z}-\frac{\gamma
   }{z}-\frac{\delta }{z-1}-1\right)e^{z-z_0}-\frac{q-\alpha  \beta  z}{(z-1) z
   (z-t)}.
 \end{align*}
\end{corollary}
~\\[-1em]

\begin{example}[Elementary integral series converging to a general Heun function]
In order to illustrate concretely the above corollary, consider the following General Heun equation (here with arbitrary parameters),
 \begin{equation}\label{ex:Hg1}
\frac{d^2H_G(z)}{dz^2} + \bigg[ \frac{2}{z} + \frac{7}{z - 1} + \frac{(-1)}{z - 4} \bigg]\frac{dH_G(z)}{dz} + \frac{(3/2)z - 1}{z(z-1)(z-4)}H_G(z) = 0,
\end{equation}
with initial conditions $H_G(6)=H'_T(6)=1$. Here, the largest real interval containing $6$ and none of the singular points $0$, $1$ and $t=4$ is $I=]4,+\infty[$.
Thus Corollary~\ref{GenHeunCorrolary} indicates that for any $z\in]4,+\infty[$, 
\begin{align*}
H_G(z) &= 1+\int_{z_0}^z G_{1}(\zeta,z_0)d\zeta,\\
&=1+\sum_{n=1}^\infty \int_{6}^z K_{1}^{\ast n}(\zeta,6)d\zeta,\\
&=1+\int_{6}^z K_{1}(\zeta,6)d\zeta+\int_{6}^z\int_{6}^{\zeta} K_{1}(\zeta,\zeta_1)K_{1}(\zeta_1,6)d\zeta_1d\zeta+\cdots,
\end{align*}
with the kernel $K_1$ given by
\begin{align*}
K_{1}(z,z_0)
&=1- e^{-z}\frac{(z-4)}{z^2(z-1)^7}\int_{z_0}^z\!e^{\zeta_1}\frac{\zeta _1\left(\zeta _1-1\right)^6}{2 \left(\zeta
   _1-4\right)^2} 
   \left(2 \zeta _1^3+10 \zeta _1^2-67 \zeta
   _1+14\right)d\zeta_1.
\end{align*}
In Fig.~(\ref{fig:Ex1}), we show a purely numerical evaluation of $H_G(z)$ together with analytical estimates based on the first few orders of the above series, i.e. we give $H^{(m)}_G(z) := 1+\sum_{n=1}^m \int_{6}^z K_{1}^{\ast n}(\zeta,6)d\zeta$, with $m=1,2,3$ and $m=6$. This exhibits the convergence of the Neumann series representation of the path-sum formulation of a general Heun function, as predicted by the theory. 
\begin{figure}[!h]
\centering
\includegraphics[width=1\textwidth]{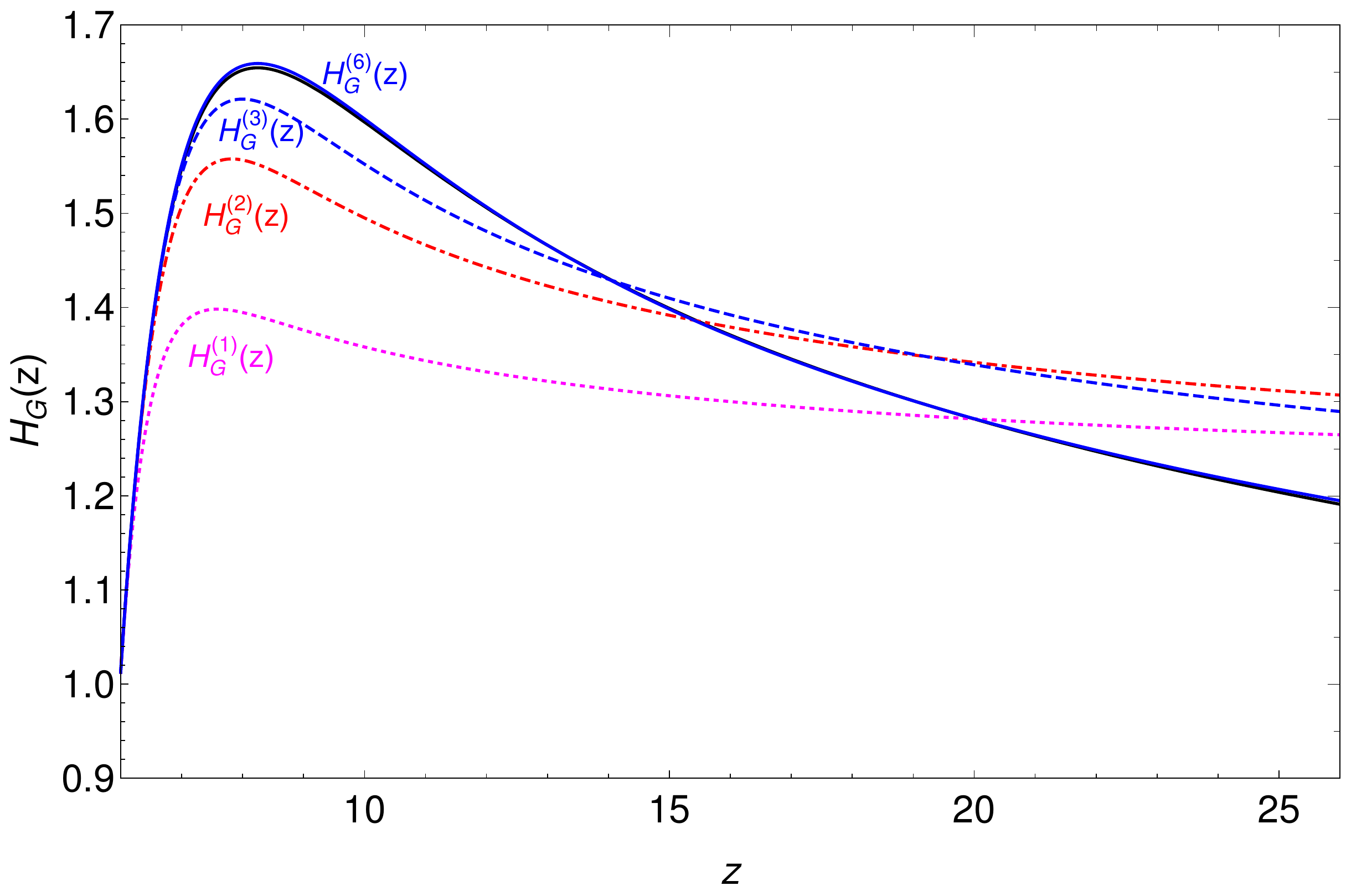}
\caption{\label{fig:Ex1}\textbf{Convergence to a general Heun function with elementary integrals.} Numerical evaluation of the general Heun function solution of Eq.~(\ref{ex:Hg1}) (solid black line), together with the first integral approximands of it: $H^{(1)}_G(z)$ (dotted magenta line), $H^{(2)}_G(z)$ (dot-dashed red line), $H^{(3)}_G(z)$ (dashed blue line) and $H^{(6)}_G(z)$ (solid blue line, very close to the numerical solution). For orders $m\geq 9$, we reach the numerical solution to within machine precision. Note that the integral series given here is convergent on $z\in ]4,+\infty[$ but we show only the interval $z\in[6,26]$ for illustration purposes.} 
\end{figure}

~\\[-1em]

The results above continue to hold should e.g. $z_0=3$, in which case $I=]1,4[$; $z_0=1/2$ implying $I=]0,1[$; or $z_0=-20$ giving $I=]-\infty,0[$. In other terms, the integral representation given for the General Heun function is valid everywhere on $z\in\mathbb{R}\backslash\{0,1,t=4\}$ but can only be used in an interval $I$ where initial conditions for $H_G$ are available.
\end{example}

\newpage
\begin{corollary}[Confluent Heun Equation]\label{corr:Conf}
Let $H_C(z)$ be solution of the Confluent Heun Equation,
 \begin{equation*}
\frac{d^2H_C(z)}{dz^2} + \bigg[ \frac{\gamma}{z} + \frac{\delta}{z - 1} + \epsilon \bigg]\frac{dH_C(z)}{dz} + \frac{\alpha z - q}{z(z-1)}H_C(z) = 0,
\end{equation*}
with initial conditions $H_C(z_0)=H_0$ and $\dot{H}_C(z_0)=H'_0$, assuming that $z_0\neq0$ and $z_0\neq 1$. If $z_0<0$, let $I=]-\infty,0[$, if $0<z_0<1$ let $I=]0,1[$, and else for $z_0>1$ let $I=]1,+\infty[$. Then, for any $z\in I$,
\begin{align*}
  H_C(z)&=H_0+H_0\int_{z_0}^z \!G_{1}(\zeta,z_0)d\zeta
+(H'_0-H_0)\left(\!e^{z-z_0}-1+\!\int_{z_0}^z\!(e^{z-\zeta}-1)G_{2}(\zeta,z_0)d\zeta\right),
  \end{align*}
where  $G_{i}=\sum_{n=1}^\infty K_{i}^{\ast n}$, $i=1,2$, and  
   \begin{align*}
K_{1}(z,z_0)&=1+ e^{-z}\int_{z_0}^z \Big\{
\frac{e^{\zeta\epsilon}\zeta^{\gamma } \left(\zeta-1\right)^{\delta}}{e^{z \epsilon} z^{\gamma } (z-1)^{\delta }} e^{\zeta}
\left(\frac{q-\alpha  \zeta}{\left(\zeta-1\right) \zeta}-\frac{\gamma }{\zeta}-\frac{\delta }{\zeta-1}-\epsilon -1\right)\Big\}d\zeta,\\
K_{2}(z,z_0)&=\left(\frac{q-\alpha z}{\left(z-1\right) z}-\frac{\gamma }{z}-\frac{\delta }{z-1}-\epsilon -1\right)e^{z-z_0}-\frac{q-\alpha  z}{(z-1) z}.
 \end{align*}
\end{corollary}
~\\[-1em]

\begin{example}[Convergence to a Confluent Heun function]
Let us now consider the following Confluent Heun function $H_C(z)$ satisfying
 \begin{equation}\label{ex:HC}
\frac{d^2H_C(z)}{dz^2} + \bigg[ \frac{3}{z} + \frac{(2/3)}{z - 1} + 4 \bigg]\frac{dH_C(z)}{dz} + \frac{5 z - 1}{z(z-1)}H_C(z) = 0
\end{equation}
with initial conditions $H_C(-5)=0$ and $H'_C(-5)=1$. Suppose that we wish to evaluate $H_C$ on the interval $z\in]-\infty,0[$, i.e. on both sides $z<z_0$ and $z>z_0$ of the conditions at $z_0=-5$. 
Then Corollary~\ref{corr:Conf} indicates that, for any $z\in]-\infty,0[$, we have 
\begin{equation*}
H_C(z) = e^{z+5}-1+\int_{-5}^z(e^{z-\zeta}-1)G_{2}(\zeta,-5)d\zeta,
\end{equation*}
with $G_2=\sum_{n=1}^\infty K_2^{\ast n}$ and
\begin{align*}
K_{2}(z,z_0)=\frac{3 (5 z-1)-e^{z-z_0} (3 z+4) (5 z-3)}{3 (z-1) z}.
\end{align*}
We emphasize that these results hold for all $z\in]-\infty,0[$ since this interval is divergence free, more precisely $K_2$ is bounded continuous on any compact subinterval of $]-\infty, 0[$ and the integral series for $G_2$ is thus guaranteed to converge on this entire domain (this is shown in the appendix). Note that when considering $z<z_0$, all integrals remain the same as for $z>z_0$.

In Fig.~(\ref{fig:ExConf}) below, we show a purely numerical evaluation of $H_C(z)$ together with the truncated integral series approximations  
\begin{align}
H_C^{(m)}(z) &:= \int_{-5}^z(e^{z-\zeta}-1)\Big(1+\sum_{n=1}^m K_{2}^{\ast n}(\zeta,-5)d\zeta\Big),\label{PSsol}\\
&=\int_{-5}^z(e^{z-\zeta}-1)\Big(1+K_{2}(\zeta,-5)+\int_{-5}^\zeta K_2(\zeta,\zeta_1)K_2(\zeta_1,-5)d\zeta_1 +\cdots\Big).\nonumber
\end{align}
 Since kernel $K_2$ is singular at $z=0$ just as $H_C$ is, we expect the convergence speed of the integral series to slow down when approaching the singular point, as predicted by the bound of Eq.(\ref{Gbound}) presented in the appendix. This does not preclude analytically obtaining the correct asymptotic behavior for $H_C(z)$ as $z\to 0^-$. Indeed this follows from the behavior of $K_2$ under the same limit. We demonstrate such a procedure in \S\ref{TRESol}.
\begin{figure}[!h]
\centering
\includegraphics[width=1\textwidth]{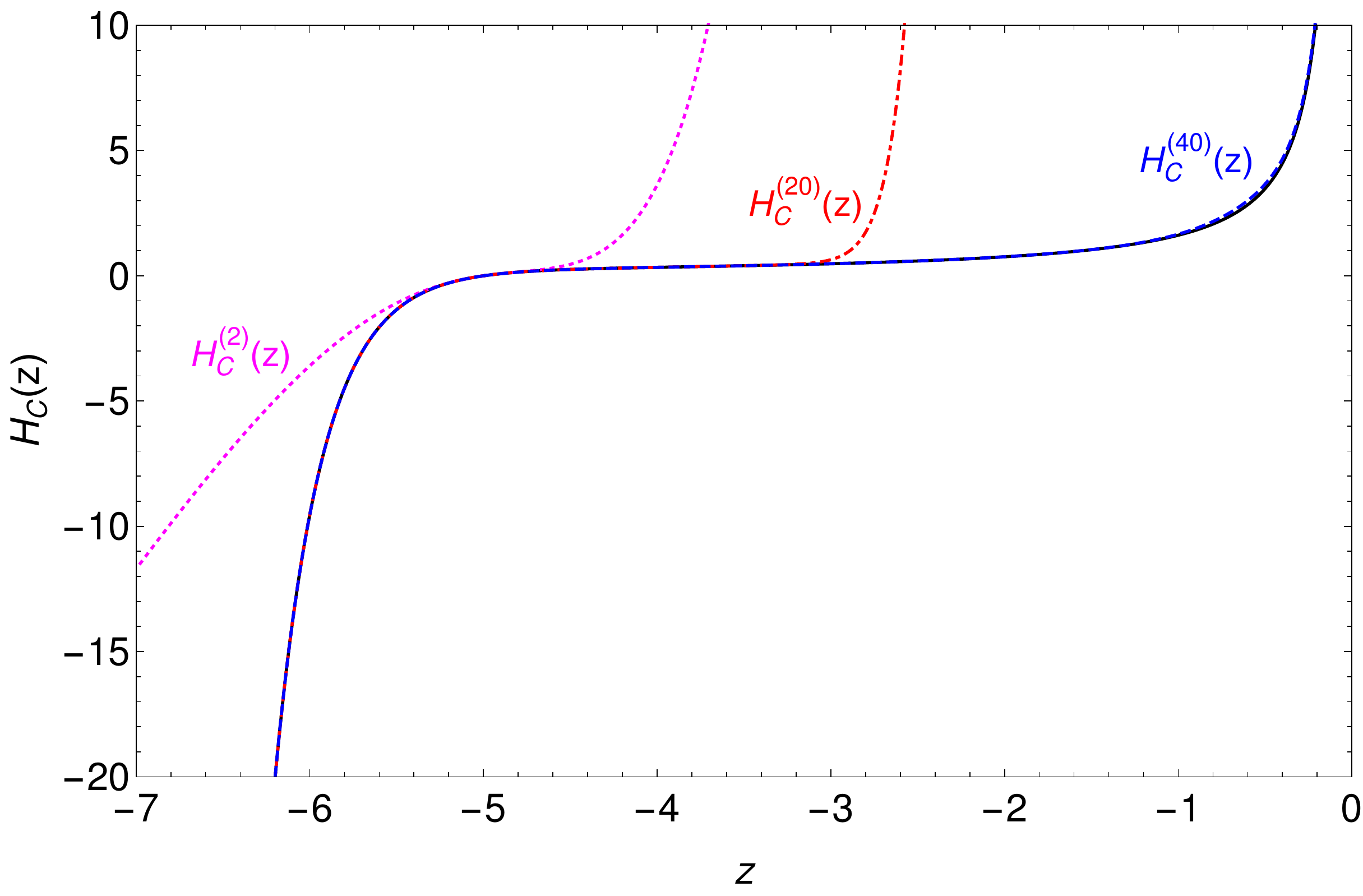}
\caption{\label{fig:ExConf}\textbf{Convergence to a Confluent Heun function with elementary integrals over the interval $]-\infty,0[$.} Numerical solution of the Eq.~(\ref{ex:HC}) (solid black line) with conditions $H_C(-5)=0$, $H_C'(-5)=1$, together with its integral approximands as per Eq.~(\ref{PSsol}), $H^{(2)}_C(z)$ (dotted magenta line), $H^{(20)}_C(z)$ (dot-dashed red line) and $H^{(40)}_C(z)$ (dashed blue line, very close to the numerical solution). Convergence near $z=0$ is slowed down due to $K_{2}$ being singular at $z=0$ just as $H_C$ is. Still, the integral series is convergent over the entire domain  $z\in]-\infty,0[$, a crucial property for perturbative black hole theory that is \emph{unique} to the present approach. Here as in subsequent examples we plot the various functions over smaller intervals for $z$, for illustration purposes.} 
\end{figure}
\end{example}
~\\
\FloatBarrier

\newpage
\begin{corollary}[Biconfluent Heun Equation]\label{corr:BiConf}
Let $H_B(z)$ be solution of the Biconfluent Heun Equation,
 \begin{equation*}
\frac{d^2H_B(z)}{dz^2} + \bigg[ \frac{\gamma}{z} + \delta + \epsilon z \bigg]\frac{dH_B(z)}{dz} + \frac{\alpha z - q}{z}H_B(z) = 0,
\end{equation*}
with initial conditions $H_B(z_0)=H_0$ and $\dot{H}_B(z_0)=H'_0$, assuming that $z_0\neq0$. If $z_0>0$, denote $I=]0,+\infty[$ otherwise let $I=]-\infty,0[$. Then, for any $z\in I$, 
\begin{align*}
  H_B(z)&=H_0+H_0\int_{z_0}^z \!G_{1}(\zeta,z_0)d\zeta
+(H'_0-H_0)\left(\!e^{z-z_0}-1+\!\int_{z_0}^z\!(e^{z-\zeta}-1)G_{2}(\zeta,z_0)d\zeta\right),
  \end{align*}
where  $G_{i}=\sum_{n=1}^\infty K_{i}^{\ast n}$, $i=1,2$, and  
   \begin{align*}
K_{1}(z,z_0)&=1+ e^{-z}\int_{z_0}^z \Big\{\frac{\zeta _1^{\gamma }}{z^{\gamma}} e^{\zeta _1-\frac{1}{2} \left(z-\zeta _1\right) \left(2 \delta +\epsilon  \left(\zeta
   _1+z\right)\right)} \Big(\frac{q-\alpha  \zeta _1}{\zeta _1}-\frac{\gamma }{\zeta _1}-\delta -\zeta _1 \epsilon -1\Big)\Big\}d\zeta_1,\\
K_{2}(z,z_0)&=\Big(\frac{q-\alpha  z}{z}-\frac{\gamma }{z}-\delta -z \epsilon -1\Big)e^{z-z_0}-\frac{q-\alpha  z}{z}.
 \end{align*}
\end{corollary}

\begin{example}[Evaluating a Biconfluent Heun function via Volterra equations]
Let us now consider the following Biconfluent Heun function $H_B(z)$ satisfying
 \begin{equation}\label{ex:HB}
 \frac{d^2H_B(z)}{dz^2} + \bigg[ \frac{(1/10)}{z} + 1 + 6 z \bigg]\frac{dH_B(z)}{dz} + \frac{(-1) z - 2}{z}H_B(z) = 0,
\end{equation}
with initial conditions $H_B(2/3)=0$ and $H'_B(2/3)=-4$. Then Corollary~\ref{corr:BiConf} indicates that for $z>0$,
\begin{align}\label{VolterraBiConf}
  H_B(z)&=2+2 \int_{2/3}^z \!\!G_{1}(\zeta,2/3)d\zeta-6\left(\!e^{z-2/3}-1+\!\int_{2/3}^z\!(e^{z-\zeta}-1)G_{2}(\zeta,2/3)\!\right) d\zeta,
  \end{align}
with $G_{i}=\sum_{n=1}^\infty K_{i}^{\ast n}$ for $i=1,2$, and
\begin{align*}
K_{1}(z,z_0)&=1+\int_{z_0}^z\frac{(19-10 \zeta  (6 \zeta +1)) e^{-(z-\zeta ) (3 \zeta +3 z+2)}}{10 \zeta ^{9/10} z^{1/10}}d\zeta,\\
K_{2}(z,z_0)&=\frac{\big(19-10 z (6 z+1)\big) e^{z-z_0}-10 (z+2)}{10 z}.
\end{align*}
Instead of evaluating functions $G_1$ and $G_2$ as the integral series, we may directly solve the linear integral Volterra equations that they satisfy, see Eq.~(\ref{VolterraG}). Such equations are very well behaved and numerically easy to solve, so that we can evaluate $H_B$ thanks to Eq.~(\ref{VolterraBiConf}) with high numerical accuracy.
In Fig.~(\ref{fig:ExBiConf}) below, we show the numerical evaluation of $H_B(z)$ obtained using a standard differential equations numerical solver versus the procedure described above.   
\begin{figure}[!h]
\centering
\includegraphics[width=.92\textwidth]{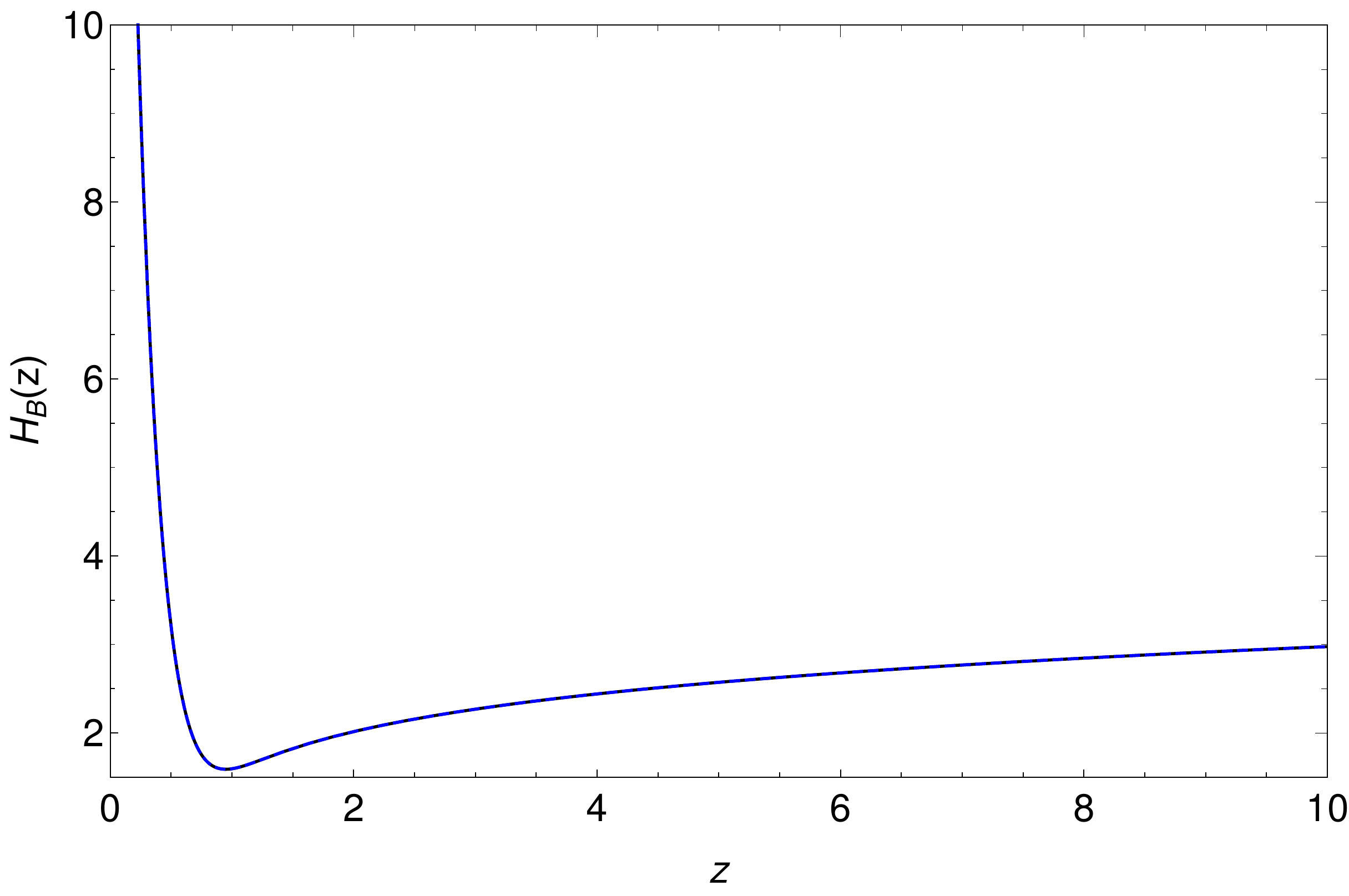}
\caption{\label{fig:ExBiConf}\textbf{Biconfluent Heun function from Volterra equations.} Numerical solution of the Eq.~(\ref{ex:HC}) (solid black line), together with the function predicted by Eq.(\ref{VolterraBiConf}) (dashed blue line). The two are indistinguishable. The integral representation of Eq.~(\ref{VolterraBiConf}) is valid for $z\in]0,+\infty[$, we here show only $z\in]0,10[$ for illustration purposes.} 
\end{figure}
\end{example}
~\\
\FloatBarrier

\begin{corollary}[Doubly-confluent Heun Equation]\label{DoublyConfluentCorr}
Let $H_D(z)$ be solution of the Doubly-confluent Heun Equation,
 \begin{equation*}
 \frac{d^2H_D(z)}{dz^2} + \bigg[ \frac{\delta}{z^2}+\frac{\gamma}{z}  + 1 \bigg]\frac{dH_D(z)}{dz} + \frac{\alpha z - q}{z^2}H_D(z) = 0
\end{equation*}
with initial conditions $H_D(z_0)=H_0$ and $\dot{H}_D(z_0)=H'_0$, assuming that $z_0\neq0$. If $z_0>0$, denote $I=]0,+\infty[$ otherwise let $I=]-\infty,0[$. Then, for any $z\in I$,  
  \begin{align*}
  H_D(z)&=H_0+H_0\int_{z_0}^z \!G_{1}(\zeta,z_0)d\zeta
+(H'_0-H_0)\left(\!e^{z-z_0}-1+\!\int_{z_0}^z\!(e^{z-\zeta}-1)G_{2}(\zeta,z_0)d\zeta\right),\,
  \end{align*}
where $G_{i}=\sum_{n=1}^\infty K_{i}^{\ast n}$, $i=1,2$, and  
   \begin{align*}
K_{1}(z,z_0)&=1+ e^{-z}\int_{z_0}^z \Big\{\frac{\zeta _1^{\gamma }}{z^{\gamma }} e^{-\frac{\delta }{\zeta _1}+2 \zeta _1+\frac{\delta }{z}-z} \Big(\frac{q-\alpha  \zeta _1}{\zeta _1^2}-\frac{\gamma }{\zeta _1}-\frac{\delta }{\zeta _1^2}-2\Big)\Big\}d\zeta_1,\\
K_{2}(z,z_0)&=\Big(\frac{q-\alpha  z}{z^2}-\frac{\delta }{z^2}-\frac{\gamma }{z}-2\Big)e^{z-z_0}-\frac{q-\alpha  z}{z^2}.
 \end{align*}
\end{corollary}
~\\[-2em]

\begin{example}[Convergence to a Doubly-Confluent Heun function]
Let us now consider the following Doubly-Confluent Heun equation, once again with arbitrarily chosen parameters for the example,
 \begin{equation}\label{ex:Hd2}
\frac{d^2H_D(z)}{dz^2} + \bigg[ \frac{(-2)}{z^2}+\frac{1}{z}  + 1 \bigg]\frac{dH_D(z)}{dz} + \frac{10 z - (-1)}{z^2}H_D(z) = 0
\end{equation}
with initial conditions $H_D(1)=H'_D(1)=1/2$. 
Then Corollary~\ref{DoublyConfluentCorr} indicates that for $z\in]0,+\infty[$,
\begin{align*}
H_D(z) &= \frac{1}{2}+\frac{1}{2}\sum_{n=1}^\infty \int_{1}^z K_{1}^{\ast n}(\zeta,1)d\zeta,
\end{align*}
with 
\begin{align*}
K_{1}(z,z_0)=1+\frac{e^{-z-\frac{1}{z}}}{z}\int_{z_0}^ze^{2\zeta_1+\frac{2}{\zeta_1}}\frac{1}{\zeta_1 }(1-2\zeta_1^2+11\zeta_1)d\zeta_1
\end{align*}
In Fig.~(\ref{fig:Ex1}) below, we show a purely numerical evaluation of $H_D(z)$ together with analytical approximations based on the first few orders of the above series, i.e. we give $H_D^{(m)}(z) := 1+\sum_{n=1}^m \int_{1}^z K_{1}^{\ast n}(\zeta,1)d\zeta$, with $m=3,5,8$. This demonstrates again the convergence of the Neumann series representation of the path-sum formulation of a general Heun function, as predicted by the theory. Here the exact $H_D(z)$ and $H_D^{(m)}(z)$ become indistinguishable for $m\geq 9$. 
\begin{figure}[h!]
\centering
\includegraphics[width=1\textwidth]{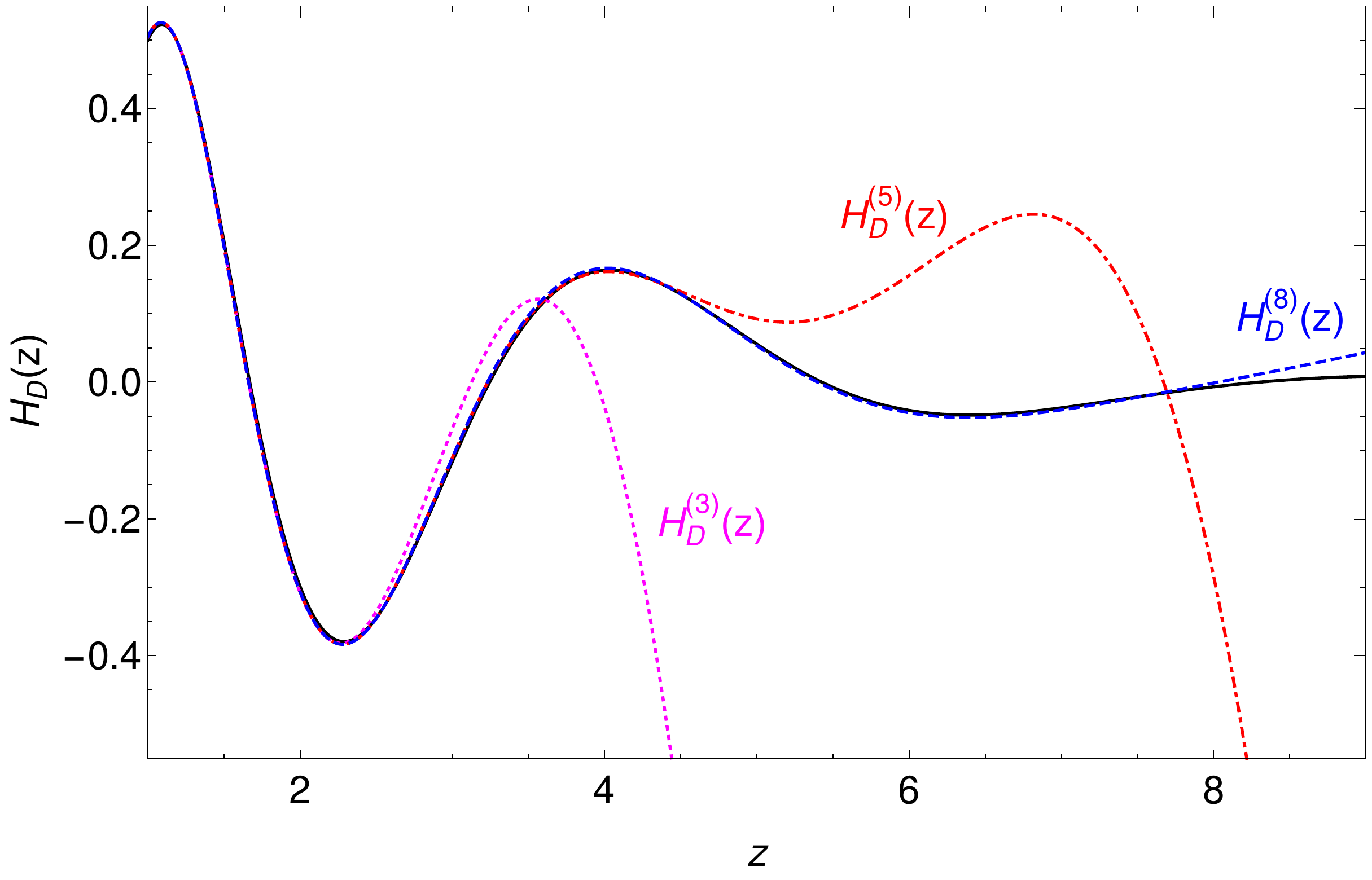}
\caption{\label{fig:Ex1}\textbf{Convergence to Doubly-Confluent Heun function with elementary integrals.} Numerical solution of the Eq.~(\ref{ex:Hd2}) (solid black line), together with its integral approximands $H^{(3)}_D(z)$ (dotted magenta line), $H^{(5)}_D(z)$ (dot-dashed red line) and $H^{(8)}_D(z)$ (dashed blue line). Note the integral series provided here is convergent for $z\in]0,+\infty[$ and show only the interval $z\in[1,9]$ for illustration purposes.} 
\end{figure}
\end{example}

\newpage
\begin{corollary}[Triconfluent Heun Equation]\label{TriconfluentCoro}
Let $H_T(z)$ be solution of the Triconfluent Heun Equation,
 \begin{equation*}
\frac{d^2H_T(z)}{dz^2} + \bigg[ \gamma + \delta z + \epsilon z^2 \bigg]\frac{dH_T(z)}{dz} + (\alpha z - q)H_T(z) = 0,
\end{equation*}
with initial conditions $H_T(z_0)=H_0$ and $\dot{H}_T(z_0)=H'_0$. Then, for any $z\in \mathbb{R}$, 
\begin{align*}
  y(z)&=H_0+H_0\int_{z_0}^z \!G_{1}(\zeta,z_0)d\zeta
+(H'_0-H_0)\left(\!e^{z-z_0}-1+\!\int_{z_0}^z\!(e^{z-\zeta}-1)G_{2}(\zeta,z_0)d\zeta\right),
  \end{align*}
where  $G_{i}=\sum_{n=1}^\infty K_{i}^{\ast n}$, $i=1,2$ and  
   \begin{align*}
K_{1}(z,z_0)&=1-e^{-\frac{1}{6} z \left(6 \gamma +2 z^2 \epsilon +3 \delta  z+6\right)}\!\!\int_{z_0}^z\!\! e^{\frac{1}{6} \zeta  \left(6 \gamma +3 \delta  \zeta +2 \zeta ^2 \epsilon +6\right)} \big(\zeta  (\alpha +\delta +\zeta 
   \epsilon )+\gamma -q+1\big) d\zeta,\\
K_{2}(z,z_0)&=-\big(\epsilon z^2+(\alpha+\delta)z -q+\gamma+1\big)e^{z-z_0}-(q-z\alpha).
 \end{align*}
\end{corollary}
~\\[-1em]

\begin{example}[Convergence to a complex-valued Triconfluent Heun function]
Consider the Triconfluent Heun function defined as the solution to
 \begin{equation}\label{ex:HT}
\frac{d^2H_T(z)}{dz^2} + \bigg[ 2  - z + 7 z^2 \bigg]\frac{dH_T(z)}{dz} + (z - (2+i))H_T(z) = 0,
\end{equation}
where $i^2=-1$, and with initial conditions $H_T(-5)=H'_T(-5)=2$. Corollary~\ref{TriconfluentCoro} indicates that for $z\in\mathbb{R}$,
\begin{align*}
H_D(z) &= 2+2 \int_{z_0}^z G_{1}(\zeta,-10)d\zeta,
\end{align*}
with $G_1=\sum_{n=1}^\infty K_1^{\ast n}$ and
\begin{align*}
K_1(z,z_0)=1-e^{\frac{1}{6}  (3-14 z) z^2-3z}\int_{z_0}^ze^{-\frac{1}{6}  (3-14 \zeta) \zeta^2+3\zeta} (7\zeta^2+1-i)\,d\zeta.
   \end{align*}
We show in Fig.~(\ref{fig:HeunT}) convergence to the complex-valued triconfluent Heun function by the integral series
\begin{align*}
H_T^{(m)}(z) &:= 2+2 \int_{z_0}^z \sum_{n=1}^m K^{\ast n}_{1}(\zeta,-10)d\zeta,\\
&=2+2\int_{z_0}^z K_{1}(\zeta,-10)d\zeta+\int_{z_0}^z \int_{z_0}^z K_{1}(z,\zeta_1)K_1(\zeta_1,-10)d\zeta_1d\zeta+\cdots
\end{align*}
With this example, we emphasize that all the integrals representations obtained here remain valid for complex-valued Heun functions. 
   \begin{figure}[h!]
\centering
\includegraphics[width=1\textwidth]{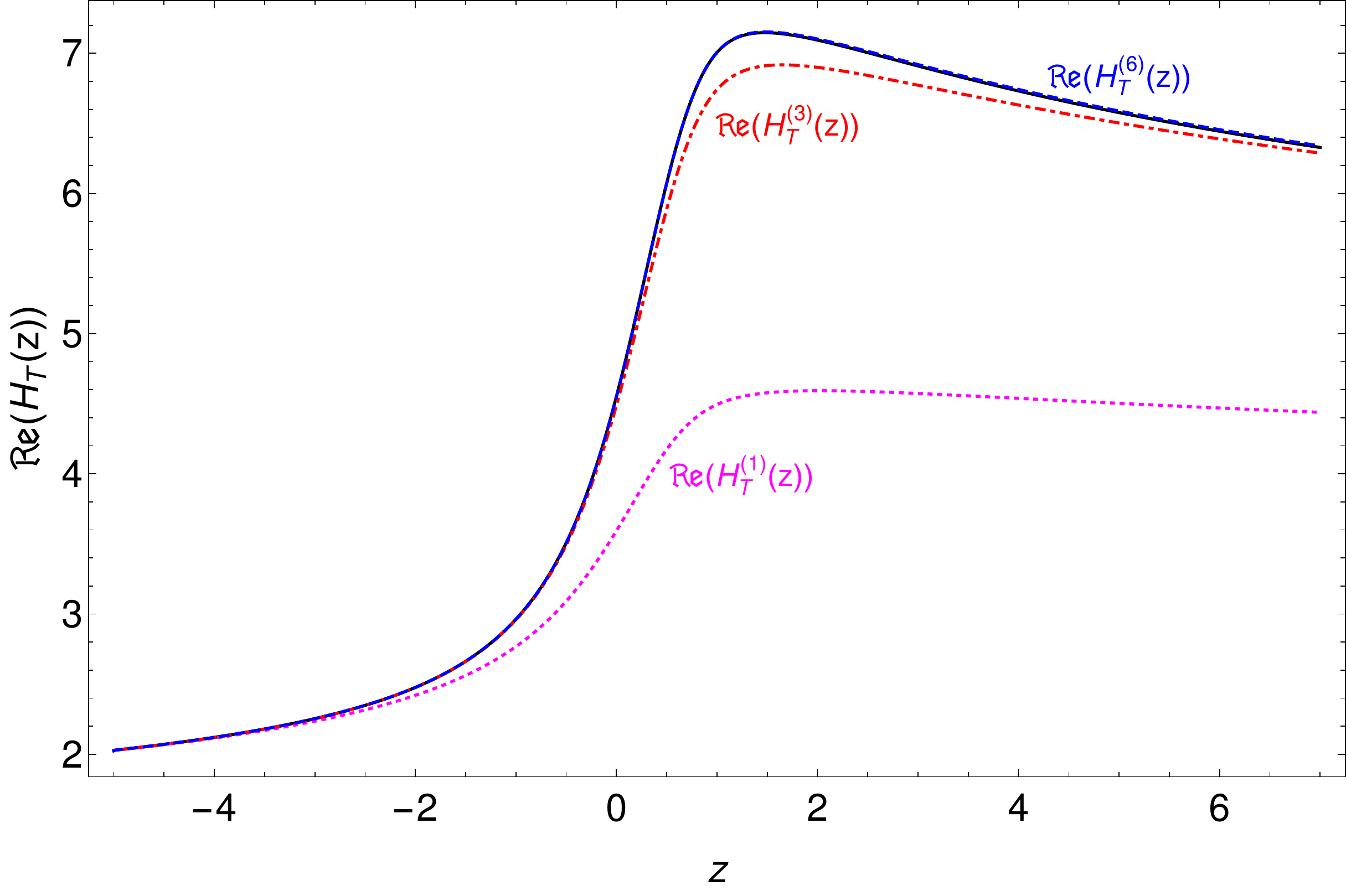}\\
\includegraphics[width=1\textwidth]{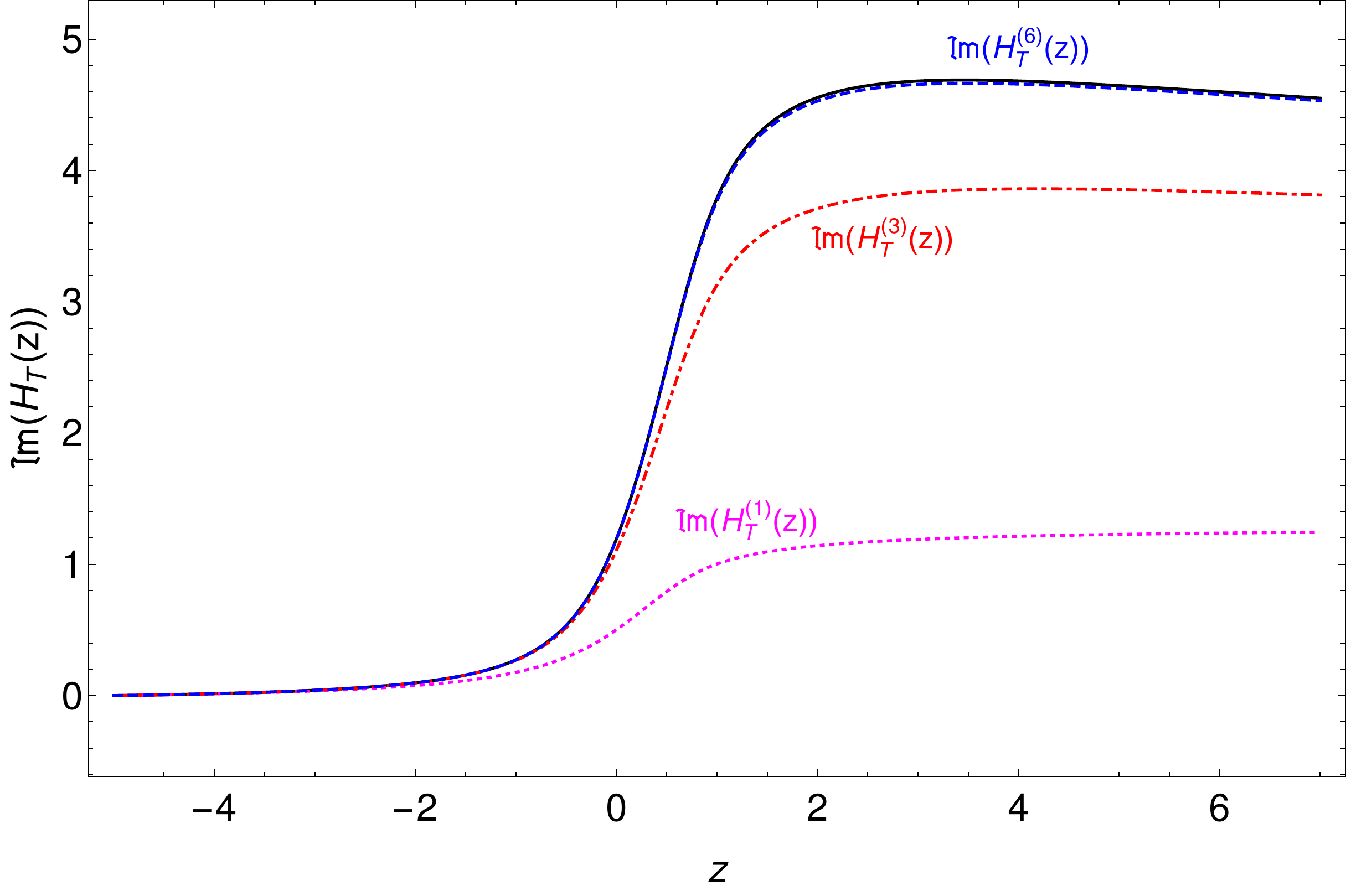}
\caption{\label{fig:HeunT}\textbf{Convergence to a complex-valued Triconfluent Heun function with elementary integrals.} Numerical solution of the Eq.~(\ref{ex:HT}) (solid black line), together with its integral approximands $H^{(1)}_T(z)$ (dotted magenta line), $H^{(3)}_T(z)$ (dot-dashed red line) and $H^{(6)}_T(z)$ (dashed blue line). Top figure: real parts of these quantities. Bottom figure: imaginary parts of these quantities. Note that the integral series provided is convergent over the entire real line, we here show only the interval $z\in [-5,7]$ for illustration purposes.} 
\end{figure}
\end{example}

 Having focused on concrete evaluations of various Heun functions in the illustrative examples, we now turn to using the elementary integral series in the field of black hole physics.

\newpage
\section{Application to Black-Hole Perturbation Theory}
\label{sec:black hole}
\subsection{Motivations}
The theory of metric perturbations of Kerr black holes is governed by the Teukolsky equation \cite{TeukEqn}. This equation provides the basic mathematical framework to study
the stability of Schwarzschild \cite{Finster2009} and Kerr black holes \cite{Costa2020} and yields physical insights in the broader field of gravitational wave astrophysics \cite{Sasaki2003}. With the advent of event detections by LIGO \cite{LIGO2016,LIGO2020}, obtaining a better analytical grasp over the solutions of the Teukolsky equation is paramount in modeling the ringdown stage \cite{London2014} of a binary black hole merger using accurate waveform templates \cite{McWilliams2019}.
\\

In the frequency domain, the Teukolsky equation can be decoupled into radial and angular components \cite{TeukEqn}. Determining the analytical solutions of the radial equation has been an active area of research since the first formulation of the equations \cite{Mano1996,Sasaki2003}. To this end, state-of-the-art approaches all rely on the same strategy: i) obtain two series expansions of the solution, one convergent near the black hole horizon the other at spatial infinity; and ii) match both expansions at some intermediate radial point.
The standard implementation of this strategy, due to Mano, Suzuki and Takasugi (MST) \cite{Mano1996,Mano1997}, relies on a series of hypergeometric functions at the black hole horizon and of Coulomb wave functions at spatial infinity. Matching both expansions requires the introduction of an auxiliary parameter $\nu$. 
We stress that this parameter is \textit{not} part of the original parameters of the Teukolsky equation. Rather $\nu$ is a mathematical checkpost introduced to establish the convergence and matching of the hypergeometric and Coulomb series \cite{Fujita2004}. The MST strategy successfully yields accurate numerical data for studying gravitational wave radiation from Kerr black holes \cite{Sasaki2003,Fujita2004}. It is \textit{``the only existing method that can be used to calculate the gravitational waves emitted to infinity to an arbitrarily high post-Newtonian order in principle.''} \cite{Sasaki2003}. At the same time, it has been explicitly recognised that the mathematical complexity of the formalism obscures physical insights into the problem \cite{Sasaki2003}. 
In particular, the auxiliary parameter $\nu$, which has been called "renormalised angular momentum" to make it more palatable, has limited correspondence to physical phenomenon, if any. 

More recently, explicit, analytic solutions to the Teukolsky equation have been established in terms of Heun functions \cite{Fiziev2}. 
Yet, Cook and Zalutskiy \cite{Cook1} note that in order to extract physical quantities of interest out of this approach, one is forced to 
revert to Leaver's formalism \cite{Leaver1985} because \textit{``the series solution around z = 1 has a radius of convergence no larger than 1, far short of infinity''}. Thus, just as for the MST formalism the problem is, in essence, that we are lacking a single representation of the solution to the Teukolsky radial equation that is convergent from the black hole horizon up to spatial infinity. 
The integral series provided in this work addresses this issue completely since it converges on this entire domain, thereby retaining the crucial features of the MST formalism that lead to its widespread applicability, while also not requiring any auxiliary, unphysical parameter. 
In a similar vein, we can assert that our formalism is suited for practical numerical and even analytical, calculations since the integral series are rapidly convergent, and their asymptotic behavior is analytically available. We may therefore also hope that the integral series representation will help solve  the well-recognised computational difficulties that emerge from the MST formalism when applied to gravitational wave physics, in particular for the two body problem \cite{Bini2013}, and in the gravitational self force program \cite{Sago2003,Hikida2004,Kavanagh2016}. 
\\

For completeness, we begin with a brief discussion of the theory of the Teukolsky equation and its reduction to Heun form. We then give the series representation of its solution. Finally, we establish its asymptotics at both the black hole horizon ($z\to 1^+$) and spatial infinity ($z\to +\infty$). 

\subsection{The Teukolsky Equation : background}
The Teukolsky Equation \cite{TeukEqn} is a gauge invariant equation \cite{GaugeInv} that governs the curvature perturbations of the Kerr black hole \cite{MWT}. By making use of the Newman-Penrose formalism \cite{NewmanPenrose}, the single master equation for the spin $(s)$ weighted scalar wave function $_{s}\psi$ in Boyer-Lindquist co-ordinates $\{t,r,\theta,\phi \}$ \cite{BoyerLindquist} and the Kinnersley tetrad \cite{Kintetrad} is written as:
\begin{gather}
\bigg[ \frac{(r^2 + a^2)^2}{ \Delta} - a^2 \sin^2 \theta \bigg] \frac{ \partial^2 _{s}\psi}{\partial t^2} +
\bigg(\frac{4Mar}{\Delta}\bigg) \frac{\partial^2 _{s}\psi}{\partial t \partial \phi}  + \bigg[ \frac{a^2}{\Delta} - \frac{1}{sin^2 \theta} \bigg] \frac{\partial^2 _{s}\psi}{\partial \phi^2} \notag \\
- \Delta^{-s} \frac{\partial}{\partial r} \bigg( \Delta^{s+1} \frac{\partial _{s}\psi}{\partial r} \bigg) - \frac{1}{\sin \theta} \frac{\partial}{\partial \theta} \bigg( \sin \theta \frac{ \partial _{s}\psi}{\partial \theta} \bigg) -2s \bigg[ \frac{a(r-M)}{\Delta} + \frac{i \cos \theta}{\sin^2 \theta} \bigg] \frac{\partial _{s}\psi}{\partial \phi} \notag \\ -2s \bigg[ \frac{M(r^2 - a^2)}{\Delta} - r - ia\cos \theta \bigg] \frac{\partial}{\partial t} + (s^2 \cot^2 \theta - s)_{s}\psi = 4 \pi \Sigma T \label {eq:TeukolskyEqn}
\end{gather}
where the auxiliary variables are given by:
\begin{gather}
\Sigma \equiv r^2 + a^2\cos^2 \theta, 
\Delta \equiv r^2 - 2Mr + a^2
\end{gather}
Here, $M$ is the mass of the black hole, $a$ is its angular momentum (per unit mass), $T$ is the source term built from the energy-momentum tensor \cite{TeukEqn} and the spin parameter $s = 0, \pm 1/2, \pm 1, \pm 2 \pm 3/2$ for scalar, neutrino, electromagnetic, gravitational and Rarita-Schwinger \cite{RaritaSchwinger} fields respectively. It reduces to the Bardeen-Press equation  \cite{BardeenPress} in the non-rotating $(a = 0)$ case.

The equation \ref{eq:TeukolskyEqn} can be separated in time \cite{Krivan1997} and frequency domain
\cite{TeukEqn}. The latter can be performed for the vacuum case $(T = 0)$ by the following  separation ansatz:
\begin{gather}
_{s}\psi(t,r,\theta, \phi) = e^{-i \omega t}e^{i m \phi}S(\theta)R(r).
\end{gather}
For the radial function $R(r)$ we obtain the Teukolsky Radial Equation (TRE): 
\begin{gather} \label{eq:TRE}
\Delta^{-s} \frac{d}{dr} \bigg[ \Delta^{s+1} \frac{d R(r)}{dr} \bigg] + \bigg[ \frac{K^2 -2is(r-M)K}{\Delta} + 4is\omega r - \lambda \bigg]R(r) = 0, 
\end{gather}
where, 
\begin{gather}
K \equiv (r^2 + a^2)\omega -am, \quad
\lambda \equiv \, _{s}A_{lm}(a\omega) + a^2\omega^2 -2am\omega.
\end{gather}
For the angular equation, we make $x \equiv \cos \theta$. Now the function  $S(\theta) = \, _{s}S_{lm}(x;a\omega)$ is the spin weighted spheroidal function \cite{Breuer1977} which gives the
solution for the Teukolsky Angular Equation (TAE):
\begin{align} \label{eq:TAE}
&\partial_{x} \bigg[(1-x^2)\partial_{x}[_{s}S_{lm}(x;c)] \bigg] + \bigg[ (cx)^2 - 2csx + s \\&\hspace{35mm}  +\, _{s}A_{lm}(c) - \frac{(m+sc)^2}{1-x^2} \bigg]  \,_{s}S_{lm}(x;c) = 0, \nonumber
\end{align}
where $c = a\omega$ is the oblateness parameter, $m$ is the azimuthal separation constant and $_{s}A_{lm}(c)$ is the angular separation constant. The equations \ref{eq:TRE} and \ref{eq:TAE} are coupled equations which require simultaneous evaluation of the parameters $\omega$ and $_{s}A_{lm}(c)$. Given a value for $_{s}A_{lm}(c)$, we can solve \ref{eq:TRE} for the complex frequency $\omega$ and given the latter, we can solve \ref{eq:TAE} as an eigenvalue problem for $_{s}A_{lm}(c)$.

\subsection{Teukolsky Radial Equation in Heun Form}
We now reduce the Teukolsky Radial Equation to the non-symmetrical Heun form, which allows us to represent its solution with the results of Section.~\ref{sec:ExplicitRes}. There is one small consideration to be noted: depending on the sign of the spin $s$ we wish to operate in, certain parameters of the CHE form of the TAE and TRE flip their signs as given in \cite{Fiziev2}. However this is not of relevance for our purposes since our main aim is to work with the CHE \textit{form} of the equations that obviously remains irrespective of the sign of the spin parameter.

The radial function $R(r)$ solution to Eq.~\ref{eq:TRE} has three singularities: an irregular singular point at $r = \infty$ and two regular singular points corresponding to the roots of $\Delta = 0$, which are
\begin{equation*}
r_{\pm} = M \pm \sqrt{M^2 - a^2}
\end{equation*}
The values $r_{\pm}$ correspond to the event and Cauchy horizon respectively (for an in-depth introduction to the notation and terminology on back-hole mathematics, we refer the reader to \cite{MWT}). Having identified these, we may now map the Teukolsky Radial Equation into an Heun equation. We close following the standard treatment \cite{Cook1}. We begin by letting the radial function $R$ be of the form  
\begin{equation}
R(r) = (r-r_{+})^{\xi}(r-r_{-})^{\eta}e^{\zeta r}H(r), \label{eq:RtoH}
\end{equation}
where the parameters $\zeta, \xi, \eta$ are given by
\begin{align} 
&\zeta = \pm i\omega \equiv \zeta_{\pm},\qquad
\xi = \frac{-2 \pm (s + 2i\sigma_{+})}{2} \equiv \xi_{\pm},\label{eq:18}\\
&\eta = \frac{ -s \pm (s - 2i\sigma_{-})}{2} \equiv \eta_{\pm},\qquad
\sigma_{\pm} = \frac{2\omega M r_{\pm} - ma}{r_{+} - r_{-}}.\nonumber
\end{align}
With the dimensionless variables
\begin{align*}
\bar{r} \equiv \frac{r}{M}, \quad
\bar{a} \equiv \frac{a}{M}, \quad
\bar{\omega} \equiv M\omega, \quad 
\bar{\zeta} \equiv M\zeta,
\end{align*}
we transform the radial coordinate $r$ into the dimensionless variable $z$ defined by
\begin{equation*}
z = \frac{r - r_{-}}{r_{+} - r_{-}} = \frac{ \bar{r} - \bar{r}_{-}}{\bar{r}_{+} - \bar{r}_{-}}
\end{equation*}
Now, any of the eight possible combinations of the parameters $\{ \zeta, \xi, \eta \}$ given in Eqs.~(\ref{eq:18}) will reduce the Teukolsky Radial Equation (\ref{eq:TRE}) into the following 
equation for the auxiliary function $H$,
\begin{equation}
\frac{d^2 H(z)}{dz^2} + \Bigg(\frac{\gamma}{z} + \frac{\delta}{z-1} +4p\Bigg)\frac{d H(z)}{dz} + \frac{4\alpha pz - \sigma}{z(z-1)}H(z) = 0 \label{eq:HeunH}
\end{equation}
which is a Confluent Heun equation. Here, the following variables have been introduced to clarify the equation,
\begin{align}
&p = (\bar{r}_{+} - \bar{r}_{-})\frac{\bar{\zeta}}{2}, \qquad \hspace{-2mm}
\alpha = 1 + s + \xi + \eta -2\bar{\zeta} + s\frac{i\bar{\omega}}{\bar{\zeta}}, \label{Param}\\
&\gamma = 1 + s + 2\eta, \qquad
\delta = 1 + s + 2\xi, \nonumber\\
&\sigma = \,_{s}A_{lm}(\bar{a}\bar{\omega}) + \bar{a}^2\bar{\omega}^2 -8\bar{\omega}^2 + p(2\alpha + \gamma - \delta) + \bigg(1 + s - \frac{\gamma + \delta}{2} \bigg) \bigg( s + \frac{\gamma + \delta}{2} \bigg).\nonumber
\end{align}
Furthermore, the local solutions at the singularities have the exact same form for all eight combinations of the parameters $\{ \zeta, \xi, \eta \}$ given in Eqs.~(\ref{eq:18}). More precisely, we get
\begin{subequations}
\begin{align}
&\lim_{z \to 0} R(z) \sim z^{-s+i\sigma_{-}} \quad \text{or} \quad z^{-i\sigma_{-}}, \label{eq:locsol1}\\
&\lim_{z \to 1} R(z) \sim (z-1)^{-s-i\sigma_{+}} \quad \text{or} \quad (z-1)^{i\sigma_{+}},\label{eq:locsol2}\\
&\lim_{z \to \infty} R(z) \sim z^{-1-2s+2i\bar{\omega}}e^{i(\bar{r}_{+} - \bar{r}_{-})\bar{\omega}z} \quad \text{or} \quad z^{-1-2i\bar{\omega}}e^{-i(\bar{r}_{+} - \bar{r}_{-})\bar{\omega}z} \label{eq:locsol3}
\end{align}
\end{subequations}
Now, the above forms correspond to behaviour of the perturbations at the boundary conditions of the event and Cauchy horizon and spatial infinity. By suitable choice of the signs in \ref{eq:locsol1},\, \ref{eq:locsol2} and \ref{eq:locsol3}, we can obtain expressions for quantities of physical interest such as Quasinormal Modes and Totally Transmitting Modes \cite{Cook1}. Also, see \cite{Navaes2019,Suzuki1998,Suzuki1999,Yoshida2010} for applications of the Heun form of the Teukolsky equations. 
The equation can be solved by various methods such as Frobenius series about the singular points \cite{Fiziev2}  and continued fractions \cite{Leaver1985}.

\subsection{Representation of the Teukolsky radial function convergent on $]1,+\infty[$}\label{TRESol}
\subsubsection{Elementary integral series}
The solution of the Confluent Heun equation~(\ref{eq:HeunH}) satisfied by the auxiliary function $H(z)$ is described by Corollary~(\ref{corr:Conf}). Since the singular points are located at $z=0,1,+\infty$, given any initial conditions for $H(z_0)$ and $\dot{H}(z_0)$ at $z_0\in]1,+\infty[$, the integral series representation of $H(z)$ is guaranteed to converge on the entire domain $]1,+\infty[$. This crucial property stands in stark contrast with the hypergeometric and Coulomb series, which converge close to 1 and to $+\infty$, respectively. Because of this, we do not need to introduce the unphysical parameter $\nu$.

Recall that the Teukolsky radial function $R$ and auxiliary function $H$ are related by Eq.~(\ref{eq:RtoH}). The auxiliary function is a confluent Heun function given by the following integral series representation, convergent for any $z\in]1,+\infty[$,
\begin{align*}
  H(z)&=H_0+H_0\int_{z_0}^z G_{1}(\zeta,z_0)d\zeta
+(H'_0-H_0)\!\left(\!e^{z-z_0}-1\!+\!\int_{z_0}^z\!\!(e^{z-\zeta}-1)G_{2}(\zeta,z_0)d\zeta\right)\!,
  \end{align*}
where  $G_{i}=\sum_{n=1}^\infty K_{i}^{\ast n}$, $i=1,2$, and  
   \begin{align*}
K_{1}(z,z_0)&=1+ e^{-(1+4p) z} z^{-\gamma } (z-1)^{-\delta }\times\\
&\hspace{5mm}\int_{z_0}^z \Big\{
e^{(1+4p)\zeta}\zeta^{\gamma } \left(\zeta-1\right)^{\delta} 
\left(\frac{\sigma-4\alpha p  \zeta}{\left(\zeta-1\right) \zeta}-\frac{\gamma }{\zeta}-\frac{\delta }{\zeta-1}-4p -1\right)\Big\}d\zeta,\\
K_{2}(z,z_0)&=\left(\frac{\sigma-4\alpha p z}{\left(z-1\right) z}-\frac{\gamma }{z}-\frac{\delta }{z-1}-4p -1\right)e^{z-z_0}-\frac{\sigma-4\alpha p  z}{(z-1) z}.
 \end{align*}
Here we assumed $z_0\in]1,+\infty[$ then $H_0:=H(z_0)$, $H'_0:=\dot{H}(z_0)$ and all parameters are given by Eq.~(\ref{Param}).\\

Witnessing to the fact that the above representation is convergent for all $z\in]1,+\infty[$, we here recover the asymptotic behavior of $H(z)$ in both limits $z\to 1^+$ and $z\to +\infty$. We emphasize that this is not possible with any single series representation of $H(z)$, which converges either in the vicinity of $1^+$ or of $+\infty$.

\subsubsection{Asymptotic behavior for $z\to +\infty$} 
From now on, we write $F(z)\sim_{a.e.}$ to present the leading term of the asymptotic expansion of the function $F(z)$, disregarding constant factors. For example, we would write $1+2/z\sim_{a.e} z^{-1}$ as $z\to 0$.\\


We begin by determining the asymptotic behavior of $K_1(z,z_0)$ for $z\gg 1$.  This depends on two cases: $p=0$ and $p\neq 0$. We suppose first that $p=0$ and assume that $\delta+\gamma>0$. In this situation, the confluent Heun function becomes a well understood hypergeometric function \cite{Erdelyi1955,Motygin2018OnEO} for which we will nonetheless show that we recover the correct asymptotic behavior. Setting $p=0$ we get, as $z\to +\infty$, 
\begin{align*}
K_1(z,z_0)&\sim_{a.e}1+ e^{-(1+4p) z} z^{-\gamma } (z-1)^{-\delta }\left(-e^{z } z ^{\gamma +\delta }+e^{z_0} z_0^{\gamma +\delta }\right),\\
&\sim_{a.e}e^{-z}z^{-\delta-\gamma}e^{z_0}z_0^{\gamma +\delta }.
\end{align*}
Then $K_1(z,z_0)$ is asymptotically the product of a function depending only on $z$ and of a function depending only on $z_0$. This property is sufficient to determine the asymptotic behavior of $G_1$ in closed-form \footnote{This is because the solution of a linear Volterra integral equation of the second kind with kernel  $K_1(z,z_0)=k(z)l(z_0)$ is known exactly \cite{giscardvolterra}.}
$$
G_1(z,z_0)\sim_{a.e}e^{-z}z^{-\delta-\gamma}e^{z_0}z_0^{\gamma +\delta }e^{\int_{z_0}^ze^{-\zeta}\zeta^{-\delta-\gamma}e^{\zeta}\zeta^{\gamma +\delta }d\zeta}=(z_0/z)^{\delta+\gamma}
$$
implying that $\int_{z_0}^z G_1(\zeta,z_0)d\zeta\sim_{a.e.} z^{1-\delta-\gamma}$ for $z\to+\infty$. Analyzing $K_2$ and $G_2$ yields the same results. Indeed, with $p=0$, we have
$$
K_2(z,z_0)\sim_{a.e.}\left(\frac{-1}{z}(\gamma+\delta)-1\right)e^{z-z_0},
$$
which is the product of a function of $z$ and a function $z_0$ so we determine
$$
G_2(z,z_0)\sim_{a.e.}\left(\frac{-1}{z}(\gamma+\delta)-1\right)e^{z-z_0}e^{-(z-z_0)}(z_0/z)^{\gamma+\delta},
$$
that is $G_2\sim_{a.e.}z^{-\gamma-\delta}$. From there $e^{z-z_0}-1\!+\!\int_{z_0}^z\!\!(e^{z-\zeta}-1)G_{2}(\zeta,z_0)d\zeta\sim_{a.e.}z^{1-\gamma-\delta}$. Thus for $p=0$ and $\delta+\gamma>0$, we get $H(z)\sim_{a.e.}z^{1-\delta-\gamma}$ regardless of the conditions at $z_0$ and provided $\delta+\gamma>0$, as expected \cite{Erdelyi1955}. Further cases arise for $\delta+\gamma\leq 0$ but we do not discuss these here as 
they correspond to well known hypergeometric results.\\

Let us now suppose that $p\neq 0$. Then, since
$$
e^{4p\zeta}\zeta^{\gamma } \left(\zeta-1\right)^{\delta} f(\zeta)=-e^{4p\zeta}\zeta^{\gamma +\delta}\left(\frac{4\alpha p+\gamma+\delta}{\zeta}+4p+O(1/\zeta^2)\right),$$
we have, asymptotically for $z\to +\infty$,
$$
K_1(z,z_0)\sim_{a.e.}1-e^{-4p z}z^{-\gamma}z^{-\delta}\times z^{\gamma +\delta } \left(e^{4 p z}-4 \alpha  p \,E_{-\gamma -\delta +1}(-4 p z)\right).
$$
where $E_{n}(z)$ is the exponential integral function, with asymptotic expansion $E_{n}(x)\sim_{a.e.}e^{-x}/x$ as $x\to+\infty$. This result  greatly simplifies $K_1$, reducing it to 
$$
K_1(z,z_0)\sim_{a.e}-\frac{\alpha}{z},\text{ as }z\to+\infty. 
$$
This allows us to determine the asymptotic behavior of $G_1$ straightforwardly as 
$$
G_1(z,z_0)\sim_{a.e}\frac{-\alpha}{z}e^{\int_{z_0}^z -\alpha /\zeta\, d\zeta}= -\alpha z^{-1-\alpha},
$$
and therefore $\int_{z_0}^z G_1(\zeta,z_0)d\zeta\sim_{a.e.} z^{-\alpha}$ for $z\to+\infty$.\\

We proceed similarly for $K_2$ and $G_2$. We have
$
K_2(z,z_0)=f(z)e^{z-z_0}+O(1/z),
$
so that asymptotically $K_2(z,z_0)\sim_{a.e.} f(z) e^z e^{-z_0}$ for $z\to+\infty$. Then $K_2(z,z_0)$ is asymptotically the product of a function depending only on $z$ and of a function depending only on $z_0$. We therefore obtain
$$
G_2(z,z_0)\sim_{a.e.} e^{z-z_0}f(z) e^{\int_{z_0}^zf(\zeta)d\zeta},\text{ as } z\to+\infty.
$$
The right hand-side is 
\begin{align*}
e^{z-z_0}f(z) e^{\int_{z_0}^zf(\zeta)d\zeta}&=e^{-4 p (z-z_0)}\left(\frac{z_0-1}{z-1}\right)^{\delta} \left(\frac{z_0}{z}\right)^{\gamma
   +\sigma} \left(\frac{1-z_0}{1-z}\right)^{4 \alpha  p-\sigma } \\&\hspace{-5mm}\times\frac{1}{z(z-1)}(\gamma -z (\gamma +\delta +4 p (\alpha
   +z-1)+z-1)+\sigma )
\end{align*} 
which yields the asymptotic result,
$$
G_2(z,z_0)\sim_{a.e.} e^{-4p z}z^{-\delta-\gamma-4\alpha p},\text{ as }z\to+\infty.
$$
This implies that 
$$
\left(\!e^{z-z_0}-1\!+\!\int_{z_0}^z\!\!(e^{z-\zeta}-1)G_{2}(\zeta,z_0)d\zeta\right)\sim_{a.e.}
e^{-4pz} z^{-4\alpha p -\delta-\gamma}.
$$
Gathering our results, we conclude that when $p\neq 0$,
\begin{align*}
&H(z)\sim_{a.e.} z^{-\alpha}~~~\text{or}~~~H(z)\sim_{a.e.} e^{-4pz} z^{-4\alpha p -\delta-\gamma},~\text{ as }z\to+\infty,
\end{align*}
which gives the same asymptotic behavior as obtained from series designed to converge when $z\to+\infty$ \cite{Motygin2018OnEO,Cook1,Ronveaux1995}. The result for $p=0$ yields the correct asymptotics of the hypergeometric function obtained in this case. 

\subsubsection{Asymptotic behavior for $z\to 1^+$} 
In this situation, we begin with 
\begin{align*}
K_{1}(z,z_0)&\sim_{a.e}1+ e^{-(1+4p) z} (z-1)^{-\delta }\int_{z_0}^z \Big\{
e^{(1+4p)\zeta} \left(\zeta-1\right)^{\delta-1} 
\left(c+c'(\zeta-1)\right)\Big\}d\zeta,
\end{align*}
where $c=\sigma-4\alpha p - \delta$ and $c'=\gamma+4p+1$.
In order to progress without presenting cumbersome equations, denote $F(\zeta)$ the following indefinite integral 
\begin{align*}
F_\delta(\zeta):&=\int e^{(1+4p)\zeta} \left(\zeta-1\right)^{\delta}d\zeta,\\
&=-e^{4 p+1} (\zeta-1)^{\delta +1} E_{-\delta }\big(-(1+4 p) (\zeta-1)\big),
\end{align*}
where $E_{n}(x)$ is the exponential integral function. In particular $E_{n}(x)\sim_{a.e.}x^{n-1}c_1+c_2$ as $x\to 0^+$ and where $c_1$ and $c_2$ are non-zero real constants that are irrelevant here. This implies $F_\delta(\zeta)\sim_{a.e.}(\zeta-1)^{1+\delta}$.
Now given that
$$
K_1(z,z_0)=1+ e^{-(1+4p) z} (z-1)^{-\delta }\big(c F_{\delta-1}(z)-cF_{\delta-1}(z_0)+c'F_{\delta}(z)-c'F_{\delta}(z_0)\big).
$$
then
$$
K_1(z,z_0)\sim_{a.e.}1,\text{ as }z\to1^+.
$$
This implies that $G_1(z,z_0)\sim_{a.e} e^{z-z_0}$ and therefore $\int_{z_0}^z G_1(\zeta,z_0)d\zeta \sim_{a.e.} 1$ as $z\to 1^+$.

For $K_2$ and $G_2$ we begin by noting that for $z$ close to 1, 
$$
K_2(z,z_0)\sim_{a.e.}\frac{4\alpha p}{z-1}(1-e^{z-z_0})-\frac{\delta}{z-1},\text{ as } z\to 1^+.
$$
from which it follows that 
$
G_2(z,z_0)\sim_{a.e.}(z-1)^{-\delta}
$ for $z\to 1^+$,  and therefore
$$
\left(\!e^{z-z_0}-1\!+\!\int_{z_0}^z\!\!(e^{z-\zeta}-1)G_{2}(\zeta,z_0)d\zeta\right)\sim_{a.e.}(z-1)^{1-\delta},\text{ as }z\to 1^+.
$$
Note that this assumes that $\delta>0$. If this is not the case, then the asymptotics is $O(1)$.

Gathering our results, we get that 
$$
H(z)\sim_{a.e} 1~\text{ or }~(z-1)^{1-\delta},\text{ as }z\to 1^+.
$$
which gives the same asymptotic behavior as obtained from series representations of $H(z)$ for $z$ close to $1$ \cite{Motygin2018OnEO,Cook1,Ronveaux1995}.

\subsection{Remarks on the Teukolsky Angular Equation}
The Teukolsky angular equation \ref{eq:TAE} has two regular singular points at $x = \pm 1$ and an irregular singular point at infinity. Just like the radial equation, we can transform it to either the Bocher symmetrical form \cite{Cook1} or the non-symmetric canonical form of the confluent Heun equation \cite{Fiziev2}. It follows that any solution to the angular equation has an integral series representation as described in this work.

The radial and angular Teukolsky equations are coupled equations, as shown e.g. by the presence of the frequency parameter $\omega$ and of the angular eigenvalue $_{s}A_{lm}$ in both the angular and radial equations. Therefore, when it comes to determining physical quantities of interest, such as quasinormal modes, the two equations must be  solved \textit{simultaneously} (we refer the reader to  \cite{Berti2006,Fiziev2,Leaver1985,Fackrell1977,Hughes2000} for methods to that end). While using integral series to solve both the radial and angular equations separately and then match the solutions is feasible, a truly ambitious alternative approach would be solve the coupled system directly with the path-sum formalism. Indeed, natively this formalism was designed to solve systems of coupled (differential) equations with variable coefficients. So much so that in order to solve the Heun equations and get an integral series representation from path-sum, the first step (see Appendix~\ref{AppendixProofs}) is to map any Heun equation back onto a \emph{system} of coupled differential equations. We believe such an approach to be feasible not only for the system comprising the angular and radial Teukolsky equations, but also for the underlying pair of coupled equations in the Penrose-Newmann formalism from  which Teukolsky obtained his equation \cite{TeukEqn}. This is beyond the scope of this work.

\section{Conclusion}\label{sec:conclusion}
In this work, we present novel integral series representations for all functions of Heun class. The major advantage of these representations is that 1) they involve only elementary integrands (rational and exponential functions); 2) they are unconditionally convergent everywhere except at the singular points of the Heun function being studied; and 3) they demonstrate that all functions of Heun class can be obtained from one or at most two Volterra equations of the second kind. Points 1) and 2) above are crucial in order to obtain physically well-behaved solutions of the homogenous Teukolsky radial
equation by means of Heun functions, as this necessitates a series representation that is convergent from the black hole horizon up to spatial infinity.
This is not feasible with state-of-the-art techniques involving hypergeometric and Coulomb series representations of confluent Heun function. The former is convergent \textit{only} near the horizon while the later is convergent \textit{only} at spatial infinity. In order to match both representations of the solutions, a book-keeping unphysical parameter $\nu$ has to be introduced which, at the very least, obscures the physical picture. 
Unlike the above MST strategy, the integral series proposed here converge over the entire spatial domain from the horizon up to infinity, thus bypassing the need for parameters that are not already present in the Teukolsky equation. 

While this work is devoted to establishing the well-posedness of the integral series formalism, the next obvious step is to use it to actually compute quantities of physical interest for the rapidly growing field of gravitational wave astrophysics. These include gravitational wave fluxes \cite{Fujita2004}, quasinormal modes \cite{London2014,Cook1} and totally transmitting modes \cite{Cook1}, all of which should now be accessible without Leaver's method (which suffers from numerical stability issues) nor the MST strategy. We hope that the formalism can also help resolve mathematical difficulties that arise in implementing the MST formalism in the various aspects of the two body problem in general relativity \cite{Bini2013,Kavanagh2016}. 

Finally, we stress that our novel mathematical results were obtained by applying the method of path-sum to Heun's equation. This method, relying on the algebraic combinatorics of walks on graphs, was originally designed to solve systems of coupled differential equations and compute matrix functions. While it already proved successful in the fields of quantum dynamics, matrix theory and combinatorics, we think that this work opens new venues for its use in ordinary differential equations and general relativity. In particular, path-sum is natively adapted to solve directly the system of coupled equations which, in the Penrose-Newman formalism, underlies the Teukolsky equation.

\begin{acknowledgements}
P.-L. G. is supported by the Agence Nationale de la Recherche
young researcher grant No. ANR-19-CE40-0006.  
\end{acknowledgements}

\newpage
\setcounter{section}{0}
\renewcommand{\thesection}{\Alph{section}}
\section{Appendix: Proof of the results}\label{AppendixProofs}
The method of proof is as follows: we map the Heun equation onto a system of two coupled linear first order differential equations with variable coefficients. The solution of such systems is given by a formal object called a path-ordered exponential, which we present below. Then we use the path-sum method to evaluate this path-ordered exponential. Finally we extract the desired Heun function from the path-sum solution.

\subsection{Path-ordered exponentials}
All the results are corollaries of the general purpose method of path-sum, which permits the exact calculation of path-ordered exponentials of finite variable matrices. The path-ordered exponential $\mathsf{U}(z)$ of a variable matrix $\mathsf{M}(z)$ is the unique matrix solution to the system of coupled first order ordinary linear differential equations with variable coefficients encoded by $\mathsf{M}(z)$, i.e.
\begin{equation}\label{OrdExpSys}
\frac{d}{dz}\mathsf{U}(z,z_0)=\mathsf{M}(z).\mathsf{U}(z,z_0),
\end{equation}
and such that for all $z_0$, $\mathsf{U}(z_0,z_0)=\mathsf{Id}$ is the identity matrix of relevant dimension. The solution of Eq.~(\ref{OrdExpSys}) is the path-ordered exponential $\mathsf{U}(z,z_0)$ of $\mathsf{M}$, denoted
\begin{align*}
\mathsf{U}(z,z_0)&=\mathcal{P}e^{\int_{z_0}^{z} \mathsf{M}(\zeta) d\zeta},\\
&=\mathsf{Id}+\int_{z_0}^z \mathsf{M}(\zeta_1)d\zeta_1+\frac{1}{2}\int_{z_0}^z\int_{z_0}^z \mathcal{P}\{\mathsf{M}(\zeta_2)\mathsf{M}(\zeta_1)\}d\zeta_2d\zeta_1+\cdots 
\end{align*}
where $\mathcal{P}$ the path-ordering operator,
$$
\mathcal{P}\big\{\mathsf{M}(\zeta_2)\mathsf{M}(\zeta_1)\big\}=\begin{cases}
\mathsf{M}(\zeta_2)\mathsf{M}(\zeta_1),&\text{if }\zeta_2\geq \zeta_1\\
\mathsf{M}(\zeta_1)\mathsf{M}(\zeta_2),&\text{otherwise.}                             
\end{cases}
$$
We refer the reader to \cite{Dyson1952} for the origins of this notation.\\[-1em]

Although used primarily to gain analytical understanding into the dynamics of quantum systems driven by time-dependent forces, path-sum relies solely on the algebraic combinatorics of walks on graphs that is valid 
irrespectively of the nature or size of the matrix $\mathsf{M}$. It is also only distantly related to the famous Feynman's path-integrals.
The interest here is that when calculating path-ordered exponentials, the method natively generates integral representations of the solutions. The strategy thus consists in calculating the ordered exponential of a matrix $\mathsf{M}(z)$ designed so that the solution of Eq.~(\ref{OrdExpSys}) should involve the desired Heun's function.

In order to recover an integral representation for all of Heun's functions, remark that Eqs.~(\ref{eq:Heun})--\ref{eq:HeunTriConfluent}) all take the form
\begin{equation}\label{HeunEq}
y''(z)-B_1(z)y'(z)-B_2(z)y(z)=0,
\end{equation}
We thus focus on obtaining the integral representation of the solution of Eq.~(\ref{HeunEq}) in terms of integrals involving $B_1$ and $B_2$, irrespectively of what these functions are. To this end, we begin by exihibiting a matrix $\mathsf{M}(z)$ whose path-ordered exponential involves a function solution to Eq.~(\ref{HeunEq}).

\begin{proposition}\label{MForm}
 Let $y(z)$ be a solution of Eq.~(\ref{HeunEq}) with initial conditions $y(z_0)=y_0$ and $y'(z_0)=y'_0$. Let
 \begin{equation}\label{Mz}
 \mathsf{M}(z)=\begin{pmatrix}1&1\\
 B_1(z)+B_2(z)-1&B_1(z)-1\end{pmatrix},
 \end{equation}
 and let $\mathsf{U}(z,z_0):=\mathcal{P}e^{\int_{z_0}^{z} \mathsf{M}(\zeta) d\zeta}$ be the path-ordered exponential of $\mathsf{M}$.
 Then
 $$
y(z) = y_0\mathsf{U}_{11}(z,z_0) +(y'_0-y_0)\mathsf{U}_{12}(z,z_0).
 $$
\end{proposition}

\begin{proof}
 By direct differentiation. Let $\psi(z)=(\psi_1(z),\psi_2(z)\big)^\mathrm{T}$ such that $\dot{\psi}(z)=\mathsf{M}(z).\psi(z)$. This implies
 $$
 \dot{\psi_1}=\psi_1+\psi_2,\qquad \dot{\psi_2}=(B_1+B_2-1)\psi_1+(B_1-1)\psi_2,
 $$
 where we omitted the $(z)$ arguments to alleviate the notation. Then
 $$
 \ddot{\psi_1}=\dot{\psi_1}+(B_1+B_2-1)\psi_1+(B_1-1)( \dot{\psi_1}-\psi_1),
 $$
 which is
 $$
 \ddot{\psi_1}-(B_1-1+1)\dot{\psi_1}-(B_1+B_2-1-B_1+1)\psi_1=0,
 $$
 i.e. $\ddot{\psi_1}-B_1\dot{\psi_1}-B_2\psi_1=0$. This is precisely Eq.~(\ref{HeunEq}). Now, since $\psi_1(z_0)=y_0$ is the desired initial condition, and since $\dot{\psi}(z_0)=\mathsf{M}(z_0).\psi(z_0)$, then to get $\dot{\psi}_1(z_0)=y'_0$ we must have $\psi_2(z_0)=y'_0-y_0$. From there and given that $\psi(z)=\mathsf{U}(z,z_0).\psi(z_0)$, we obtain
 $$
 \psi_1(z)=y_0\mathsf{U}_{11}(z,z_0) +(y'_0-y_0)\mathsf{U}_{12}(z,z_0),
 $$
which completes the proof.
\end{proof}

\subsection{Path-sum formulation}
We may now use the method of path-sum to calculate the path-ordered exponential of $\mathsf{M}$ to recover the desired integral representations. We first state and prove the general result concerning Eq.~(\ref{HeunEq}) before giving its corollaries in the specific cases of the general Heun, confluent, biconfluent, doubly-confluent and triconfluent Heun functions.
\begin{theorem}\label{GenTheorem}
Let $\mathsf{M}(z)$ be given as in Eq.~(\ref{Mz}), let $\mathsf{U}(z,z_0)$ be its path-ordered exponential. 
Then 
$$
\mathsf{U}_{11}(z,z_0)=1+\int_{z_0}^z G_{1}(\zeta,z_0)d\zeta,
$$
where $G_{1}(z,z_0)$ satisfies the linear integral Volterra equation of the second kind 
 $$
 G_{1}(z,z_0)=K_{1}(z,z_0)+\int_{z_0}^z K_{1}(z,\zeta) G_{1}(\zeta,z_0) d\zeta,
 $$
 with kernel
 \begin{align*}
K_{1}(z,z_0)&=1+ e^{-z}\int_{z_0}^z \Big\{e^{\zeta}\,e^{\int_{\zeta}^z B_1(\zeta')d\zeta'} (B_1(\zeta)+B_2(\zeta)-1)\Big\}d\zeta,
 \end{align*}
Furthermore, 
$$
\mathsf{U}_{12}(z,z_0)=\int_{z_0}^z(e^{z-\zeta}-1)\big(1+G_{2}(\zeta,z_0)\big) \,d\zeta,
$$
where $G_{2}(z,z_0)$ satisfies the linear integral Volterra equation of the second kind
 $$
 G_{2}(z,z_0)=K_{2}(z,z_0)+\int_{z_0}^z K_{2}(z,\zeta) G_{2}(\zeta,z_0) d\zeta,
 $$
 with kernel
 \begin{align*}
K_{2}(z,z_0)&=(B_1(z)+B_2(z)-1)e^{z-z_0}-B_2(z).
 \end{align*}
 \end{theorem}
 
 Given that a linear Volterra integral equations of the second kind always has an explicit solution in the form of a Neumann series of the kernel obtained from Picard iteration, we present below the ensuing elementary integral series representations for the solution $y_0 \mathsf{U}_{11}(z,z_0)+(y'_0-y_0)\mathsf{U}_{12}(z,z_0)$ of Eq.~(\ref{HeunEq}). This will be greatly facilitated by  Volterra compositions, presented in the proof of the Theorem.

 \begin{remark}
  If the initial conditions are such that $y_0=y'_0$, then by Proposition~\ref{MForm}, the solution $y(z)$ is directly proportional to $\mathsf{U}_{11}$. By Theorem~\ref{GenTheorem} this implies that the derivative $\dot{y}(z)$ of the solution of Eq.~(\ref{HeunEq}) satisfies a linear Volterra integral equation of the second kind with kernel $K_{1}$ given above. In other terms, any solution of any Heun equation which has at least one point $z_0$ for which $y(z_0)=y'(z_0)$ satisfies such a Volterra integral equation with kernel $K_{1}$. This is the first known integral equation satisfied by Heun functions in terms of elementary functions.
  Similarly if $y(z_0)=0$, then the solution $y(z)$ is proportional to $\mathsf{U}_{12}$, itself an integral of $G_{2}$ which satisfies a linear Volterra integral equation of the second kind. 
\end{remark}

\begin{proof}
 The central mathematical concept enabling the path-sum formulation of path-ordered exponentials is the $\ast$-product. This product is defined on a large class of distributions \cite{GiscardPozza2020}, however for the present work only its definition on smooth functions of two variables is required. For such functions the $\ast$-product reduces to the Volterra composition, a product between functions first expounded by Volterra and P\'er\`es in the 1920s \cite{Volterra1924} and which had largely fallen out of use by the early 1950s for a reason that appears, restrospectively, to be the lack of a mathematical theory of distributions. 
The Volterra composition of two smooth functions of two variables $f(z,z_0)$ and $g(z,z_0)$ is  
$$
\big(f\ast g\big)(z,z_0)=\int_{z_0}^{z} f(z,\zeta)g(\zeta,z_0) d\zeta\,\Theta(z-z_0),
$$
with $\Theta(.)$ the Heaviside theta function under the convention that $\Theta(0)=1$.
This extends to functions of less than two variables, for example if $h(z)$ is a smooth function of one variable, then
\begin{align*}
\big(h\ast g\big)(z,z_0)&=h(t')\int_{z_0}^{z} g(\zeta,z_0) d\zeta\,\Theta(z-z_0),\\
\big(g\ast h\big)(z,z_0)&=\int_{z_0}^{z} g(z,\zeta)h(\zeta) d\zeta\,\Theta(z-z_0).
\end{align*}
That is, the variable of $h(z)$ is always treated as the left variable of a function of two variables. 

The identity element for the $\ast$-product is the Dirac distribution, denoted $1_\ast\equiv \delta(z-z_0)$, an observation which we here accept without proof as it would require presenting the full theory of the $\ast$-product \cite{GiscardPozza2020}. Similarly we  accept without proof that for any bounded function $f(z,z_0)$ of two variables, $f^{\ast 0}=1_\ast$, while $f^{\ast 1}=f $ and $f^{\ast n+1}=f\ast f^{\ast n}=f^{\ast n}\ast f$ \cite{Volterra1924}.
Furthermore, if $f$ is bounded the Neumann series $\sum_{n=0}^\infty \big(f^{\ast n}\big)(z,z_0)$ converges superexponentially and thus  unconditionally \cite{Linz1985} to an object, called the $\ast$-resolvent $R_f$ of $f$, given by 
\begin{align*}
R_f(z,z_0)&=\sum_{n=0}^\infty \big(f^{\ast n}\big)(z,z_0),\\
&=\delta(z-z_0)+f(z,z_0)\Theta(z-z_0)+\int_{z_0}^z f(z,\zeta_1)f(\zeta_1,z_0)d\zeta_1\Theta(z-z_0)\\
&\hspace{10mm}+\int_{z_0}^z\int_{\zeta_1}^z f(z,\zeta_2)f(\zeta_2,\zeta_1)f(\zeta_1,z_0)d\zeta_2d\zeta_1 \Theta(z-z_0)+\cdots.
\end{align*}
Seeing this as steming from a Picard iteration entails an additional property of $\ast$-resolvents, namely that they solve the Volterra equation of the second kind with kernel $f$, 
\begin{equation}\label{VolterraR}
R_f=1_\ast + f \ast R_f,
\end{equation}
or, in explicit integral notation, and showing all distributions
$$
R_f(z,z_0)=\delta(z-z_0)+\!\!\int_{z_0}^{z}\!\!\!f(z,\zeta) R_f(\zeta,z_0)d\zeta\, \Theta(z-z_0).
$$
Thus we have $R_f\ast (1_\ast - f)= 1_\ast$ and are therefore justified in writing $R_f=\big(1_\ast-f\big)^{\ast-1}$. In order to avoid distributions altogether, it is more convenient to define $G_f:=R_f-1_\ast$ and rewite Eq.~(\ref{VolterraR}) as
\begin{equation*}
G_f=f+f\ast G_f,
\end{equation*}
which is an ordinary linear integral Volterra equation of the second kind.
The Neumann integral series obtained from Picard iterations for $G_f$ as above is now
$$
G_f=\sum_{n=1}^\infty f^{\ast n},
$$
and it is a well-established result \cite{Linz1985} that the convergence of this series is guaranteed provided $f$ is continuous and bounded. In this case, truncating the series at order $m$, yields a relative error of at most
$$
\left|G_f(z,z_0)-\sum_{n=1}^m f^{\ast n}(z,z_0)\right|\leq \frac{\kappa_f^m}{m!}
$$
with $\kappa_f:=\sup_{\zeta,\zeta'\in]z_0,z[:\,\zeta\geq z'}|f(\zeta,\zeta')|$.\\[-0em]

Path-sum expresses the path-ordered exponential of any finite variable matrix in terms of a finite number of Volterra compositions and $\ast$-resolvents. The path-sum formulation of the path-ordered exponential of the $2\times 2$ matrix $\mathsf{M}(z)$ is 
  \begin{align*}
  \mathsf{U}_{11}(z)&=1\ast R_1,\\
  R_1&=\left(1_\ast-\mathsf{M}_{11}-\mathsf{M}_{12}\ast\big(1_\ast-\mathsf{M}_{22}\big)^{\ast-1}\ast\mathsf{M}_{21}\right)^{\ast-1},
  \end{align*}
  where the $\ast$-multiplication by $1$ on the left is a short-hand notation for an integral with respect to the left variable, since for any $f$ smooth, $(1\ast f)(z,z_0)=\int_{z_0}^z f(\zeta,z_0)d\zeta\Theta(z-z_0)$. Furthermore, since $\mathsf{M}_{22}$ depends on a single variable, its $\ast$-resolvent can be shown to be 
  $$
  \big(1_\ast-\mathsf{M}_{22}\big)^{\ast-1}=1_\ast+ \mathsf{M}_{22}e^{1\ast \mathsf{M}_{22}},
  $$
  or equivalently
  $$
  \big(1_\ast-\mathsf{M}_{22}\big)^{\ast-1}(z,z_0)=\delta(z-z_0)+\mathsf{M}_{22}(z)e^{\int_{z_0}^z\mathsf{M}_{22}(\zeta)d\zeta}\Theta(z-z_0).
  $$
  Now the form of $\mathsf{U}_{11}$ as claimed in the theorem follows upon writing the $\ast$-products as explicit integrals with $\mathsf{M}$ given by Proposition.~(\ref{MForm}).
  For $\mathsf{U}_{12}$, the path-sum formulation reads
  $$
  \mathsf{U}_{12}=1\ast \big(1_\ast-\mathsf{M}_{11}\big)^{\ast-1}\ast \mathsf{M}_{12}\ast R_{2},
  $$
  where
  $$
  R_{2}=\left(1_\ast-\mathsf{M}_{22}-\mathsf{M}_{21}\ast\big(1_\ast-\mathsf{M}_{11}\big)^{\ast-1}\ast\mathsf{M}_{12}\right)^{\ast-1},
  $$
  and the theorem result for $\mathsf{U}_{12}$ follows upon writing the $\ast$-products as explicit integrals with $\mathsf{M}$ given by Proposition.~(\ref{MForm})
 \end{proof}
 
 Since $R_{1}$ and $R_{2}$ are $\ast$-resolvents, we may express them as the unconditionally convergent Neumann series involving the corresponding kernels $K_{1}$ and $K_{2}$, i.e. $R_{i}=1_\ast+\sum_{n=1}^\infty K_{i}^{\ast n}$ or equivalently $G_{i}=\sum_{n=1}^\infty K_{i}^{\ast n}$, $i=1,2$. This yields an explicit representation for the solution of Eq.~(\ref{HeunEq}) as series of elementary integrals:
 
 \begin{theorem}\label{TheoExplicit}
  Let $y(z)$ be the unique solution of
  $$
  y''(z)-B_1(z)y'(z)-B_2(z)y(z)=0,
  $$
  such that $y(z_0)=y_0$ and $y'(z_0)=y'_0$. Then
  \begin{align*}
  y(z)&=y_0+y_0\int_{z_0}^z G_{1}(\zeta,z_0)d\zeta
+(y'_0-y_0)\int_{z_0}^z(e^{z-\zeta}-1)\big(1+G_{2}(\zeta,z_0)\big)d\zeta,
  \end{align*}
  where $G_{1}$ and $G_{2}$ satisfy linear Volterra integral equations of the second kind with kernels respectively given by
   \begin{align*}
K_{1}(z,z_0)&=1+ e^{-z}\int_{z_0}^z \Big\{e^{\int_{\zeta}^z B_1(\zeta')d\zeta'}e^{\zeta} \big(B_1(\zeta)+B_2(\zeta)-1\big)\Big\}d\zeta,\\
K_{2}(z,z_0)&=(B_1(z)+B_2(z)-1)e^{z-z_0}-B_2(z).
 \end{align*}  
  In consequence, $G_{1}$ and $G_{2}$ have the following representation as  integral series inbvolving elementary integrands, 
  \begin{align}
  G_{i}(z,z_0)&=\sum_{n=1}^\infty K_{i}^{\ast n}(z,z_0),\label{Gexplicitform}\\
  &=K_{i}(z,z_0)+\int_{z_0}^z K_{i}(z,\zeta_1)K_{i}(\zeta_1,z_0)d\zeta_1\nonumber\\
  &\hspace{3mm}+\int_{z_0}^z \int_{\zeta_1}^zK_{i}(z,\zeta_2)K_{i}(\zeta_2,\zeta_1)K_{i}(\zeta_1,z_0)d\zeta_2d\zeta_1+\nonumber\\
  &\hspace{5mm}+\int_{z_0}^z \int_{\zeta_1}^z\int_{\zeta_2}^z K_{i}(z,\zeta_3)K_{i}(\zeta_3,\zeta_2)K_{i}(\zeta_2,\zeta_1)K_{i}(\zeta_1,z_0)d\zeta_3 d\zeta_2d\zeta_1+\cdots,\nonumber
  \end{align}
  for $i=1,2$.
  The series representation is guaranteed to converge to $G_i$ everywhere except at the singular points of $K_i$. More precisely, let $]z_0,z_1[$ be an open interval over which $K_i$ is divergent free and let $\kappa_i:=\sup_{\zeta,\zeta'\in ]z_0,z_1[:\, \zeta\geq \zeta'}|K_i(\zeta,\zeta')|$. Then 
  \begin{equation}\label{Gbound}
  \left|G_i(z,z_0)-\sum_{n=1}^m K_i^{\ast n}(z,z_0)\right| \leq \frac{\kappa_i^m}{m!}.
  \end{equation}
  \end{theorem}

This immediately provides the Corollaries of the main text for  the general Heun, confluent, biconfluent, doubly-confluent and triconfluent Heun's functions upon replacing $B_1$ and $B_2$ appearing in Eq.~(\ref{HeunEq}) and Theorem~\ref{TheoExplicit} with their values as dictated by Eqs.~(\ref{eq:Heun}--\ref{eq:HeunTriConfluent}). 

\providecommand{\noopsort}[1]{}\providecommand{\singleletter}[1]{#1}%


\begin{thebibliography}{10}
\providecommand{\url}[1]{{#1}}
\providecommand{\urlprefix}{URL }
\expandafter\ifx\csname urlstyle\endcsname\relax
  \providecommand{\doi}[1]{DOI~\discretionary{}{}{}#1}\else
  \providecommand{\doi}{DOI~\discretionary{}{}{}\begingroup
  \urlstyle{rm}\Url}\fi

\bibitem{LIGO2016}
Abott, B., \textit{et al.}: Observation of gravitational waves from a binary
  black hole merger.
\newblock Phys. Rev. Lett. \textbf{116}, 061,102 (2016).
\newblock \doi{10.1103/PhysRevLett.116.061102}

\bibitem{LIGO2020}
Abott, R., \textit{et al.}: G{W}190412: Observation of a binary-black-hole
  coalescence with asymmetric masses.
\newblock Phys. Rev. D. \textbf{102}, 043,015 (2020).
\newblock \doi{10.1103/PhysRevD.102.043015}

\bibitem{Mathieu3}
Amendola, G.: Application of {M}athieu functions to the analysis of radiators
  conformal to elliptic cylindrical surfaces.
\newblock Journal of Electromagnetic Waves and Applications \textbf{13}(8),
  1103--1120 (1935)

\bibitem{Mathieu2}
Antikainen, A., Essiambre, R., Agrawal, G.: Determination of modes of
  elliptical waveguides with ellipse transformation perturbation theory.
\newblock Optica \textbf{4}(12), 1510 (2017)

\bibitem{BardeenPress}
Bardeen, J., Press, W.: Radiation fields in the {S}chwarzschild background.
\newblock Jour. Math. Phys. \textbf{14}, 7--19 (1973)

\bibitem{Berti2006}
Berti, E., Cardoso, V., Casals, M.: Eigenvalues and eigenfunctions of spin-
  weighted spheroidal harmonics in four and higher dimensions.
\newblock Phys. Rev. D. \textbf{73}, 024,013 (2006).
\newblock \doi{10.1103/PhysRevD.73.024013}

\bibitem{Bini2013}
Bini, D., Damour, T.: Analytical determination of the two-body gravitational
  interaction potential at the 4th post-newtonian approximation.
\newblock Phys. Rev. D. \textbf{87}, 121,501(R) (2013).
\newblock \doi{10.1103/PhysRevD.87.121501}

\bibitem{Fiziev2}
Borissov, R., Fiziev, P.: Exact solutions of {T}eukolsky master equation with
  continuous spectrum.
\newblock Bulg. J. Phys. \textbf{48}, 065--089 (2010)

\bibitem{BoyerLindquist}
Boyer, R., Lindquist, R.: Maximal analytic extension of the {K}err metric.
\newblock Jour. Math. Phys. \textbf{8}, 265 (1967)

\bibitem{Breuer1977}
Breuer, R., Ryan, M., Waller, S.: Some properties of spin-weighted spheroidal
  harmonics.
\newblock Proc. R. Soc. Lond. A \textbf{358}, 71--86 (1977).
\newblock \doi{10.1098/rspa.1977.0187}

\bibitem{Carlitz}
Carlitz, L.: Orthogonal polynomials related to elliptic functions.
\newblock Duke Math. J. \textbf{27}, 443--459 (1960)

\bibitem{RaritaSchwinger}
Castillo, G., Ortigoza, G.: {R}arita-{S}chwinger fields in the {K}err geometry.
\newblock Phys. Rev. D. \textbf{42}, 4082 (1990)

\bibitem{Cook1}
Cook, G., Zalutskiy, M.: Gravitational perturbations of the {K}err geometry:
  High-accuracy study.
\newblock Phys. Rev. D. \textbf{90}, 124,021 (2014)

\bibitem{Costa2020}
Costa, R.: Mode stability for the teukolsky equation on extremal and
  subextremal kerr spacetimes.
\newblock Comm. Math. Phys. \textbf{378}(1), 705--781 (2020)

\bibitem{Crampe}
Cramp{\'e}, N., Nepomechie, R., Vinet, L.: Free-fermion entanglement and
  orthogonal polynomials.
\newblock J. Stat. Mech. \textbf{2019}(9), 093,101 (2019)

\bibitem{Mathieu}
Daniel, D.: Exact solutions of {M}athieu equations.
\newblock Prog. Theor. Exp. Phys. \textbf{4}, 043A01 (2020)

\bibitem{Dorey}
Dorey, P., Suzuki, J., Tateo, R.: Finite lattice {B}ethe ansatz systems and the
  {H}eun equation.
\newblock J. Phys. A: Math. Gen. \textbf{37}(6), 2047 (2004)

\bibitem{RieHibH}
Dubrovin, B., Kapaev, A.: A {R}iemann-{H}ilbert approach to the {H}eun
  equation.
\newblock SIGMA \textbf{14}, 93 (2018)

\bibitem{Dyson1952}
Dyson, F.J.: Divergence of {{Perturbation Theory}} in {{Quantum
  Electrodynamics}}.
\newblock Physical Review \textbf{85}(4), 631--632 (1952).
\newblock \doi{10.1103/PhysRev.85.631}

\bibitem{El_Jaick_2011}
El-Jaick, L.J., Figueiredo, B.D.B.: Transformations of heun's equation and its
  integral relations.
\newblock Journal of Physics A: Mathematical and Theoretical \textbf{44}(7),
  075,204 (2011).
\newblock \doi{10.1088/1751-8113/44/7/075204}.
\newblock \urlprefix\url{https://doi.org/10.1088%2F1751-8113%2F44%2F7%2F075204}

\bibitem{Erdelyi}
Erd{\'e}lyi, A.: Integral equations for {H}eun functions.
\newblock The Quarterly Journal of Mathematics \textbf{13}(1), 107--112 (1942)

\bibitem{Erdelyi1955}
Erd{\'e}lyi, A., Magnus, W., Oberhettinger, F., Tricomi, F.G.: Higher
  Transcendental Functions (Vol. 3).
\newblock McGraw Hill (1955)

\bibitem{Fackrell1977}
Fendley, E., Crossman, R.: Spin‐weighted angular spheroidal functions.
\newblock Jour. Math. Phys. \textbf{18}, 1849 (1977)

\bibitem{GaugeInv}
Fernandes, J., Lun, A.: Gauge invariant perturbations of black holes. ii.
  {K}err space-time.
\newblock Jour. Math. Phys. \textbf{38}(1), 330--349 (1997)

\bibitem{Finster2009}
Finster, F., Smoller, J.: Decay of solutions of the teukolsky equation for
  higher spin in the schwarzschild geometry.
\newblock Adv. Theor. Math. Phys. \textbf{13}(1), 71--110 (2009)

\bibitem{Fujita2004}
Fujita, R., Tagoshi, H.: New numerical methods to evaluate homogeneous
  solutions of the teukolsky equation.
\newblock Prog. Theor. Phys. \textbf{112}, 1079--1096 (2004).
\newblock \doi{10.1143/PTP.112.415}

\bibitem{QuantComp}
Giorgadze, G.: Monodromic quantum computing.
\newblock International Journal for Computer Research \textbf{15}(3), 259--294
  (2009)

\bibitem{giscardvolterra}
Giscard: On the solutions of linear volterra equations of the second kind with
  sum kernels.
\newblock J. Integral Equations Applications  (2020).
\newblock \urlprefix\url{https://projecteuclid.org:443/euclid.jiea/1581649215}.
\newblock Advance publication

\bibitem{Giscard2015}
Giscard, P.L., Lui, K., Thwaite, S.J., Jaksch, D.: An exact formulation of the
  time-ordered exponential using path-sums.
\newblock Journal of Mathematical Physics \textbf{56}(5), 053,503 (2015).
\newblock \doi{10.1063/1.4920925}.
\newblock \urlprefix\url{https://doi.org/10.1063/1.4920925}

\bibitem{GiscardPozza2020}
Giscard, P.L., Pozza, S.: Tridiagonalization of systems of coupled linear
  differential equations with variable coefficients by a lanczos-like method
  (2020)

\bibitem{Heun1888}
Heun, K.: {Zur Theorie der Riemann'schen Functionen zweiter Ordnung mit vier
  Verzweigungspunkten}.
\newblock Mathematische Annalen \textbf{33}, 161--179 (1888).
\newblock \doi{10.1007/BF01443849}.
\newblock \urlprefix\url{https://doi.org/10.1007/BF01443849}

\bibitem{HeunOrig}
Heun, K.: Zur theorie der {R}iemann’chen functionen zweiter ordnung mit
  verzweigungspunkten.
\newblock Math. Ann. \textbf{33}, 161--179 (1889)

\bibitem{Hikida2004}
Hikida, W., Nakano, H., Sasaki, M.: Self-force regularization in the
  schwarzschild spacetime.
\newblock Class. Quant. Grav. \textbf{22}, S753 (2004).
\newblock \doi{10.1088/0264-9381/22/15/009}

\bibitem{Mathieu1}
Holland, R., Cable, V.: Mathieu functions and their applications to scattering
  by a coated strip.
\newblock IEEE Transactions on Electromagnetic Compatibility \textbf{34}(1),
  9--16 (1992)

\bibitem{Hort1}
Hortaçsu, M.: Heun functions and some of their applications in physics.
\newblock Analytical Methods for High Energy Physics \textbf{2018}, 8621,573
  (2018)

\bibitem{Hort2}
Hortaçsu, M., Birkandan, T.: Quantum field theory applications of {H}eun type
  functions.
\newblock Rep. Math. Phys. \textbf{79}(1), 81--87 (2017)

\bibitem{Hughes2000}
Hughes, S.: Evolution of circular, nonequatorial orbits of {K}err black holes
  due to gravitational-wave emission.
\newblock Phys. Rev. D. \textbf{61}, 084,004 (2000)

\bibitem{Ishkhanyan_2005}
Ishkhanyan, A.: Incomplete beta-function expansions of the solutions to the
  confluent heun equation.
\newblock Journal of Physics A: Mathematical and General \textbf{38}(28),
  L491--L498 (2005).
\newblock \doi{10.1088/0305-4470/38/28/l02}.
\newblock \urlprefix\url{https://doi.org/10.1088%2F0305-4470%2F38%2F28%2Fl02}

\bibitem{QTS2}
Ishkhanyan, A., Shahverdyan, T., Ishkhanyan, T.: Thirty five classes of
  solutions of the quantum time-dependent two-state problem in terms of the
  general {H}eun functions.
\newblock Eur. Phys. J. D \textbf{69}, 10 (2015)

\bibitem{Kavanagh2016}
Kavanagh, C., Ottewill, A., Wardell, B.: Analytical high-order post-newtonian
  expansions for spinning extreme mass ratio binaries.
\newblock Phys. Rev. D. \textbf{93}, 124,038 (2016).
\newblock \doi{10.1103/PhysRevD.93.124038}

\bibitem{Slav3}
Kazakov, A., Slavyanov, S.: {E}uler integral symmetries for the confluent
  {H}eun equation and symmetries of the {P}ainleve equation.
\newblock Theor. Math. Phys. \textbf{179}(2), 543--549 (2014)

\bibitem{Kintetrad}
Kinnersley, W.: Type {D} vacuum metrics.
\newblock Jour. Math. Phys. \textbf{10}, 1195 (1969)

\bibitem{Krivan1997}
Krivan, W., Laguna, P., Papadopoulos, P., Andersson, N.: Dynamics of
  perturbations of rotating black holes.
\newblock Phys. Rev. D \textbf{56}, 3395 (1995).
\newblock \doi{10.1103/PhysRevD.56.3395}

\bibitem{ConfPot}
Leaute, B., Marcilhacy, G.: On the {S}chr{\"o}dinger equations of rotating
  harmonic, three-dimensional and doubly anharmonic oscillators and a class of
  confinement potentials in connection with the biconfluent {H}eun differential
  equation.
\newblock J. Phys. A: Math. Gen. \textbf{19}(17), 3527--3533 (1986)

\bibitem{Leaver1985}
Leaver, E.: An analytic representation for the quasi-normal modes of {K}err
  black holes.
\newblock Proc. R. Soc. Lond. A \textbf{402}, 285--298 (1985).
\newblock \doi{10.1098/rspa.1985.0119}

\bibitem{Linz1985}
Linz, P.: Analytical and Numerical Methods for Volterra Equations.
\newblock Society for Industrial and Applied Mathematics (1985).
\newblock \doi{10.1137/1.9781611970852}

\bibitem{London2014}
London, L., Shoemaker, D., Healy, J.: Modeling ringdown: Beyond the fundamental
  quasinormal modes.
\newblock Phys. Rev. D. \textbf{90}, 124,032 (2014).
\newblock \doi{10.1103/PhysRevD.90.124032}

\bibitem{Maier1}
Maier, R.: On reducing the {H}eun equation to hypergeometric equation.
\newblock J. Diff. Equ. \textbf{213}, 171--203 (2004)

\bibitem{Maier192}
Maier, R.: The 192 solutions of the {H}eun equation.
\newblock Math. Comp. \textbf{76}, 811--843 (2007)

\bibitem{Maier2}
Maier, R.: Special Functions and Orthogonal Polynomials.
\newblock American Mathematical Society, USA (2008)

\bibitem{Mano1997}
Mano, S., Eiichi, T.: Analytic solutions of the teukolsky equation and their
  properties.
\newblock Prog. Theor. Phys. \textbf{97}, 213--232 (1997).
\newblock \doi{10.1143/PTP.97.213}

\bibitem{Mano1996}
Mano, S., Suzuki, H., Takasugi, E.: Analytical solutions of the teukolsky
  equation and their low frequency expansions.
\newblock Prog. Theor. Phys. \textbf{95}, 1079--1096 (1996).
\newblock \doi{10.1143/PTP.95.1079}

\bibitem{IntPot}
Marcilhacy, G., Pons, R.: The {S}chr{\"o}dinger equation for the interaction
  potential $x^2 \lambda x^2/(1+gx^2)$ and the first {H}eun confluent equation.
\newblock J. Phys. A: Math. Gen. \textbf{18}(13), 2441--2449 (1985)

\bibitem{Mathieu5}
McLanchlan, N.: Theory and Applications of Mathieu Functions.
\newblock Oxford University Press, London (1947)

\bibitem{McWilliams2019}
McWilliams, S.: Analytical black-hole binary merger waveforms.
\newblock Phys. Rev. Lett. \textbf{122}, 191,102 (2019).
\newblock \doi{10.1103/PhysRevLett.122.191102}

\bibitem{MWT}
Misner, C., Thorne, K., Wheeler: Gravitation.
\newblock W.H. Freeman and Company, San Francisco, USA (1973)

\bibitem{Moham}
Mohamadian, T., Negro, J., Nieto, L., Panahi, H.: Tavis-{C}ummings models and
  their quasi-exactly solvable {S}chr{\"o}dinger {H}amiltonians.
\newblock Eur. Phys. J. Plus \textbf{134}, 363 (2019)

\bibitem{Motygin2018OnEO}
Motygin, O.V.: On evaluation of the confluent heun functions.
\newblock 2018 Days on Diffraction (DD) pp. 223--229 (2018)

\bibitem{NewmanPenrose}
Newman, E., Penrose, R.: An approach to gravitational radiation by a method of
  spin coefficients.
\newblock Jour. Math. Phys. \textbf{3}, 566 (1962)

\bibitem{Navaes2019}
Novaes, F., Marinho, C., Lencs{\'e}s, Casals, M.: {K}err-de sitter quasinormal
  modes via accessory parameter expansion.
\newblock J. High Energ. Phys. p.~33 (2019).
\newblock \doi{10.1007/JHEP05(2019)033}

\bibitem{Mathieu4}
Paul, W.: Electromagnetic traps for charged and neutral particles.
\newblock Rev. Mod. Phys. \textbf{62}, 531 (1990)

\bibitem{Ronveaux1995}
Ronveaux, A., Arscott, F.M., Slavyanov, S.Y., D., S., G., W., P., M., A., D.:
  Heun's Differential Equations.
\newblock Oxford University Press (1995)

\bibitem{Mathieu6}
Ruby, L.: Applicaitions of the {M}athieu equation.
\newblock Am. J. Phys. \textbf{64}, 39 (1996)

\bibitem{Sago2003}
Sago, N., Nakano, H., Sasaki, M.: Gauge problem in the gravitational
  self-force: Harmonic gauge approach in the schwarzschild background.
\newblock Phys. Rev. D. \textbf{67}, 104,017 (2003).
\newblock \doi{10.1103/PhysRevD.67.104017}

\bibitem{Sasaki2003}
Sasaki, M., Tagoshi, H.: Analytic black hole perturbation approach to
  gravitational radiation.
\newblock Living Rev. Relativ. \textbf{6}, 6 (2003).
\newblock \doi{10.12942/lrr-2003-6}

\bibitem{QTS1}
Shahverdyan, T., Ishkhanyan, T., Grigoryan, A., Ishkhanyan, A.: Analytic
  solutions of the quantum two-state problem in terms of the double bi- and
  triconfluent {H}eun functions.
\newblock J. Contemp. Physics (Armenian Ac. Sci.) \textbf{50}, 211--226 (2015)

\bibitem{Slav5}
Slavyanov, S.: Asymptotic Solutions of the One-dimensional {S}chr{\"o}dinger
  Equation.
\newblock American Mathematical Society Translation of Mathematical Monographs
  151 (1996)

\bibitem{Slav1}
Slavyanov, S.: Antiquantization and the corresponding symmetries.
\newblock Theor. Math. Phys. \textbf{182}, 1522--1526 (2015)

\bibitem{Slav4}
Slavyanov, S.: {P}ainlev{\'e} equations as classical analogues of {H}eun
  equation.
\newblock J. Phys. A. \textbf{29}, 7329--7335 (2018)

\bibitem{SpecFunc}
Slavyanov, S., Lay, W.: Special Functions: A Unified Theory Based on
  Singularities.
\newblock Oxford University Press, New York (2000)

\bibitem{Slav2}
Slavyanov, S., Salatich, A.: Confluent {H}eun equation and confluent
  hypergeometric equation.
\newblock J. Math. Sci. \textbf{232}(2), 157--163 (2018)

\bibitem{Sleeman}
Sleeman, B.: Non-linear {I}ntegral equations for {H}eun functions.
\newblock Proceedings of the Edinburgh Mathematical Society \textbf{16}(4),
  281--289 (1969)

\bibitem{Suzuki1999}
Suzuki, H., Takasugi, E., Umetsu, H.: Analytic solutions of the {T}eukolsky
  equation in {K}err-de {S}itter and {K}err-{N}ewman-de {S}itter geometries.
\newblock Prog. Theor. Phys. \textbf{102}(2), 253--272 (1998).
\newblock \doi{10.1143/PTP.102.253}

\bibitem{Suzuki1998}
Suzuki, H., Takasugi, E., Umetsu, H.: Perturbations of {K}err-de {S}itter black
  holes and {H}eun's equations.
\newblock Prog. Theor. Phys. \textbf{100}(3), 491--505 (1998).
\newblock \doi{10.1143/PTP.100.491}

\bibitem{Takem1}
Takemura, K.: Heun equation and {P}ainlev{\'e} equation.
\newblock Workshop: Studies on elliptic integrable system  (2004)

\bibitem{Takem2}
Takemura, K.: Integral transformation of {H}eun equations and some
  applications.
\newblock J. Math. Soc. Japan \textbf{69}(2), 849--891 (2017)

\bibitem{TeukEqn}
Teukolsky, S.: Perturbations of a rotating black hole. i. fundamental equations
  for gravitational, electromagnetic, and neutrino-field perturbations.
\newblock Astroph. J. \textbf{185}, 635--648 (1973)

\bibitem{Valent}
Valent, G.: An {I}ntegral transform involving {H}eun functions and a related
  eigenvalue problem.
\newblock SIAM J. Math. Anal. \textbf{17}(3), 688--703 (1986)

\bibitem{BelyiHeun}
Vidunas, R., Filipuk, G.: A classification of coverings yielding
  {H}eun-to-hypergeometric reductions.
\newblock Osaka J, Math. \textbf{51}(3), 867--703 (2014)

\bibitem{Volterra1924}
Volterra, V., P\'er\`es, J.: Le\c{c}ons sur la composition et les fonctions
  permutables.
\newblock \'Editions Jacques Gabay (1924)

\bibitem{Whittaker}
Whittaker, E., Watson, G.: A Course of Modern Analysis.
\newblock Cambridge University Press, United Kingdom (1915)

\bibitem{Xie2010}
Xie, Q., Hai, W.: Analytical results for a monochromatically driven two-level
  system.
\newblock Phys. Rev. A. \textbf{82}(3), 032,117 (2010)

\bibitem{Yoshida2010}
Yoshida, S., Uchikata, N., Futamese, T.: Quasinormal modes of {K}err–de
  {S}itter black holes.
\newblock Phys. Rev. D. \textbf{81}, 044,005 (2010).
\newblock \doi{10.1103/PhysRevD.81.044005}

\bibitem{Zhang2013}
Zhang, Z., Berti, E., Cardoso, V.: Quasinormal ringing of kerr black holes. ii.
  excitation by particles falling radially with arbitrary energy.
\newblock Phys. Rev. D \textbf{88}, 044,018 (2013).
\newblock \doi{10.1103/PhysRevD.88.044018}

\end{thebibliography}
\end{document}